\documentclass[journal,twoside,web]{ieeecolor}

\usepackage{generic}
\usepackage{cite}
\usepackage{algorithmic}
\usepackage{graphicx}
\usepackage{subfig}
\usepackage{textcomp}
\def\BibTeX{{\rm B\kern-.05em{\sc i\kern-.025em b}\kern-.08em
    T\kern-.1667em\lower.7ex\hbox{E}\kern-.125emX}}
\usepackage{microtype}
\usepackage{epsfig}
\usepackage{epstopdf}
\usepackage[cmex10]{amsmath}
\usepackage{mathtools}
\mathtoolsset{showonlyrefs=true}
\usepackage{amssymb}
\usepackage{amsfonts}
\usepackage[normalem]{ulem}
\usepackage{url}
\usepackage{hyperref}

\usepackage[ruled]{algorithm2e}

\usepackage{amsthm}

\usepackage[usenames,dvipsnames]{xcolor}

\newcommand{\vect}[1]{\text{vec}\text{$(#1)$}}
\newcommand{\rank}[1]{\text{rank}\text{$(#1)$}}

\newcommand{\F}[1]{\text{${#1}$}}
\newcommand{\E}[1]{\text{${#1}$}}

\theoremstyle{definition}
\newtheorem{theorem}{Theorem}
\newtheorem{proposition}{Proposition}
\newtheorem{corollary}{Corollary}
\newtheorem{lemma}{Lemma}
\newtheorem{example}{Example}
\newtheorem{assumption}{Assumption}
\newtheorem{definition}{Definition}
\newtheorem{problem}{Problem}

\theoremstyle{remark}
\newtheorem{remark}{Remark}

\newcommand{\mygamma}{\Psi}

\begin{document}
\title{Output Regulation of Linear Stochastic Systems}
\author{Alberto Mellone, \IEEEmembership{Student Member, IEEE}, and Giordano Scarciotti, \IEEEmembership{Member, IEEE}
\thanks{Alberto Mellone and Giordano Scarciotti are with the Department of Electrical and Electronic \mbox{Engineering}, Imperial College London, London, SW7 2AZ, UK. {\tt\small \{a.mellone18,g.scarciotti\}@imperial.ac.uk}}}

This article has been accepted for publication by IEEE Transactions on Automatic Control.

\vspace{1cm}

The manuscript included in this file is the open access accepted version. 

\vspace{1cm}

This open access version is released on arXiv in accordance with the IEEE copyright agreement.

\vspace{1cm}

The final version will be available at (not open access) \url{https://doi.org/10.1109/TAC.2021.3064829}

\maketitle

\begin{abstract}
We address the output regulation problem for a general class of linear stochastic systems. Specifically, we formulate and solve the \emph{ideal} full-information and output-feedback problems, obtaining perfect, but non-causal, asymptotic regulation. A characterisation of the problem solvability is deduced. We point out that the ideal problems cannot be solved in practice because they unrealistically require that the Brownian motion affecting the system is available for feedback. Drawing from the ideal solution, we formulate and solve \emph{approximate} versions of the full-information and output-feedback problems, which do not yield perfect asymptotic tracking but can be solved in a realistic scenario. These solutions rely on two key ideas: first we introduce a discrete-time \emph{a-posteriori} estimator of the variations of the Brownian motion obtained causally by sampling the system state or output; second we introduce a hybrid state observer and a hybrid regulator scheme which employ the estimated Brownian variations. The approximate solution tends to the ideal as the sampling period tends to zero. The proposed theory is validated by the regulation of a circuit subject to electromagnetic noise.
\end{abstract}


\section{Introduction}
The problem of output regulation is fundamental in control theory. It consists in designing a controller for a system such that the closed-loop system attains the properties of stability and regulation. In particular, the former is meant as the convergence of the state of the closed-loop system to zero in absence of external inputs. The latter requires the output of the system to track reference signals and/or to reject disturbances. Both references and disturbances are assumed to be generated by a so-called exogenous system (sometimes referred to as exosystem or signal generator).

Since the problem was firstly addressed and solved for linear deterministic systems by Davison, Francis and Wonham, \cite{Davison1976}, \cite{Francis1977}, \cite{FrancisWonham1976} and \cite{Wonham1985}, it has been widely studied and research has followed different directions: the problem has been extended to and solved for the nonlinear counterpart in~\cite{IsidoriByrnes1990} and \cite{HuangRugh1990}; additional results have been achieved in more general nonlinear frameworks in, \emph{e.g.}, \cite{Byrnes1997}, \cite{Serrani2001}, \cite{Byrnes2003}, \cite{Byrnes2004}, \cite{Huang2004}, \cite{Pavlov2006}, \cite{Marconi2007}, \cite{Marconi2008}; finally, the problem has been further extended to other classes of dynamical systems, such as time-varying \cite{Ichikawa2006}, hybrid \cite{Marconi2013}, \cite{Carnevale2016} or networked \cite{Davison1976bis}, \cite{Wieland2011}.

In this paper we deal with the problem of output regulation of a class of stochastic systems described by stochastic differential equations. This modeling framework allows the designer to account for uncertainties which have probabilistic properties that the classical robust deterministic framework cannot address. A series of examples motivating the need for the use of stochastic modeling in control engineering are provided in \cite[Section 1.9]{Damm2004} and all the references therein. These examples include mechanical systems (the inverted pendulum, the two-cart system, the quarter-car model, the car-steering system), aerospace applications (the satellite dynamics along the pitch and yaw directions), electrical and electro-mechanical systems (the suspended gyro) and mathematical finance. Other methodological applications include the filtering and optimal control problems which are addressed in \emph{e.g.} \cite{Yong1999} and \cite{Oksendal2003}, and the stochastic $H_\infty$ control problem, see, \emph{e.g.}, \cite{Hinrichsen1998}, \cite{Gershon2013} and \cite{Hua2018}.

Although control of stochastic systems is a longstanding research area, a systematic theory of output regulation is missing. Some results in this direction have been obtained in \cite{He2013}, \cite{He2014}; therein randomness is induced by the system dynamics switching in a Markovian fashion among underlying linear deterministic subsystems. However, stochastic differential equations are not involved.

This paper is intended to solve the problem of output regulation for single-input single-output systems described by general (state, control and reference/disturbance signals appear in both the drift and diffusion terms) linear stochastic differential equations and, in doing so, lay the foundations to further extend the results beyond the linear case. To achieve these goals we formulate and solve four regulation problems: \emph{ideal full information}, \emph{approximate full information}, \emph{ideal output feedback} and \emph{approximate output feedback}. The ideal problems require zero tracking error at steady state but their solutions are not causal, whereas the approximate problems allow for a non-zero tracking error but are causal.

We first deal with the full-information problem, \emph{i.e.} we assume that the state of the system is available for feedback. In this regard, we first provide the solution as well as necessary and sufficient conditions for the solvability of the problem in the ideal, but unrealistic, case in which the Brownian motion affecting the system is available for feedback. This hypothesis is not practically reasonable and this motivates the name \emph{ideal} which we give to this framework. We then show that the ideal solution is instrumental in introducing and solving an \emph{approximate}, yet practically sound, problem. In particular, as the noise is not available, the reference tracking requirement cannot be met perfectly. Thus, we propose a feedback control scheme which allows us to trade accuracy of regulation for practical implementation. To do so, we introduce a hybrid architecture that periodically performs an \emph{a-posteriori} partial estimation of the noise that affected the system between sampling times. The full-information problem is then revisited and an approximate yet implementable solution is provided. Most importantly, we show that the tracking error decreases as the sampling frequency is increased. In an analogous way, we solve the ideal output-feedback problem, \emph{i.e.} by assuming knowledge of the Brownian motion, we use a measured output to synthesise a dynamic regulator achieving perfect state estimation and asymptotic tracking. Again, this ideal, but unrealistic, framework is instrumental in formulating and solving an approximate counterpart. In this case, the Brownian motion is partially estimated by comparing successive samples of an additional measured output. This \emph{a-posteriori} partial estimation of the noise is used to design a causal state observer and hence a causal dynamic regulator achieving approximate reference tracking. We show that, if samples are acquired with higher frequency, both the state estimation error and the tracking error decrease.

Preliminary results have been published in \cite{Scarciotti2018}, \cite{Mellone2019} and \cite{Mellone2019bis}. The additional contributions of this paper are both theoretical and practical: 1) all the results are now proved, providing a substantial theoretical contribution. 2) Differently from the preliminary publications, all the results are now independent of a specific feedback gain. This is a major improvement, which also reconciles the stochastic framework with the deterministic output regulation. 3) The solvability of the regulator equations, as well as the solvability of the ideal problems, is characterised and shown to be a generalisation of the deterministic non-resonance condition. 4) The results are generalized to systems with any relative degree. 5) The design of a regulator achieving the stochastic internal model property is provided. 6) All the results are hereby revisited and reorganised in an organic and systematic way. 7) Finally, we present a practical example that validates and justifies the development of the theory.

The rest of the paper is organised as follows. In Section~\ref{sec:general_framework} we set up the framework and report some preliminary notions. In Section~\ref{section:ideal_fi} we state the ideal full-information problem and provide its solution. We additionally characterise the solvability of the problem and discuss the stochastic internal model property. In Section~\ref{section:approximate_fi} we state an approximate version of the full-information problem. We then show how to obtain an \emph{a-posteriori} approximation of the variations of the Brownian motion and we use this to synthesise a hybrid regulator that solves the problem. In Section~\ref{section:ideal_ef} we state the problem of output regulation via output feedback in the ideal case and report its solution. In Section~\ref{section:approximate_ef} we discuss additional challenges in the implementation of the ideal dynamic compensator, which justify the statement of an approximate problem. A new way of obtaining estimates of the variation of the Brownian motion is introduced; then, a hybrid dynamic regulator that solves the approximate problem is designed. In Section \ref{sec:example} an example is provided to illustrate the results. Finally, Section \ref{sec:conclusions} contains some concluding remarks. 

\textbf{Notation.} The symbol $\mathbb{Z}$ denotes the set of integer numbers, while $\mathbb{R}$ and $\mathbb{C}$ denote the fields of real and complex numbers, respectively; by adding the subscript ``$<\!0$'' (``$\ge \!0$'', ``$0$'') to any symbol indicating a set of numbers, we denote that subset of numbers with negative (non-negative, zero) real part. By $\Vert a \Vert$ we denote the Euclidean norm of $a\in \mathbb{R}^{n}$. If a function $g$ experiences a jump variation at time $\bar t$, we use the notation $g(\bar{t}^+) = \lim_{t\rightarrow \bar{t}^+}g(t)$ to indicate the value of $g$ after the jump, while we indicate the limit from the left, \emph{i.e.} $\lim_{t\rightarrow \bar{t}^-}g(t)$, simply by $g(\bar t)$. The identity matrix is denoted by~$I$. $A^\top$ indicates the transpose of $A$. The symbol $\otimes$ denotes the Kronecker product. $(\nabla, \mathcal{A}, \mathcal{P})$ is a probability space given by the set $\nabla$, the $\sigma$-algebra $\mathcal{A}$ defined on $\nabla$ and the probability measure $\mathcal{P}$ on the measurable space $(\nabla, \mathcal{A})$. $E[X]$ denotes the expected value of the random variable $X:\nabla \rightarrow (\nabla, \mathcal{A})$ \cite[Section 1.3]{Arnold1974}. A \textit{stochastic process} with state space $\mathbb{R}^n$ is a family $\{x_t,\,t\in\mathbb{R}\}$ of $\mathbb{R}^n$-valued random variables, \textit{i.e.} for every fixed $t\in\mathbb{R}$, $x_t(\cdot)$ is an $\mathbb{R}^n$-valued random variable and, for every fixed $w \in \nabla$, $x_{\cdot}(w)$ is an $\mathbb{R}^n$-valued function of time \cite[Section 1.8]{Arnold1974}. For ease of notation, we often indicate a stochastic process $\{x_t,\,t\in\mathbb{R}\}$ simply with $x_t$ (this is common in the literature, see \textit{e.g.} \cite{Arnold1974}). A stochastic process with the superscript $ss$ denotes the steady state of such process. The symbol $\mathcal{W}_t$ indicates a standard Wiener process, also referred to as Brownian motion, defined on the probability space $(\nabla, \mathcal{A}, \mathcal{P})$ \cite[Chapter 3]{Arnold1974}. The stochastic integrals, and thus the solution of stochastic differential equations, are meant in It\^o's sense.

\section{Preliminaries}
\label{sec:general_framework}
Consider the linear stochastic single-input single-output system
\begin{equation}
\label{eq:general_system}
\begin{aligned}
dx_t\! &=\! (Ax_t\!+\!Bu_t\! +\! P\omega)dt\! +\! (Fx_t\! +\! Gu_t\! +\! R\omega)d\mathcal{W}_t,\\
y^c_t &= Cx_t + Du_t,\quad e_t = y^c_t + Q\omega,
\end{aligned}
\end{equation}
where  $x_t \in \mathbb{R}^n$ is the state, $u_t\in \mathbb{R}$ is the control input, $\omega(t) \in \mathbb{R}^{\nu}$ is the exogenous input, $y^c_t\in \mathbb{R}$ is the controlled output, $e_t\in \mathbb{R}$ is the tracking error, $A\in \mathbb{R}^{n\times n}$, $F \in \mathbb{R}^{n\times n}$, $B\in \mathbb{R}^{n\times 1}$, $G \in \mathbb{R}^{n\times 1}$, $P\in \mathbb{R}^{n\times \nu}$, $R \in \mathbb{R}^{n\times \nu}$, $C\in \mathbb{R}^{1\times n}$, $D\in \mathbb{R}$ and $Q\in\mathbb{R}^{1\times \nu}$. For notational simplicity we assume that the initial condition $x_0$ is deterministic. The~exogenous signal $\omega$ is assumed to be the state of a so-called exogenous system, which is described by the equations
\begin{equation}
\label{eq:exogenous_system}
\dot{\omega} = S\omega, \quad \omega(0) = \omega_0,
\end{equation}
where $S\in \mathbb{R}^{\nu\times\nu}$ and the initial condition $\omega_0$ is assumed deterministic for notational simplicity. Moreover, we often recall and assume the following.

\begin{assumption}
	\label{assumption:ex_system_marg_stable}
	All the eigenvalues of the matrix $S$ are on the imaginary axis and are simple, \emph{i.e.} they have the same algebraic and geometric multiplicity.
\end{assumption}
By Assumption~\ref{assumption:ex_system_marg_stable} and independence of $x_0$ from $\mathcal{W}_t - \mathcal{W}_0$ for all $t\in \mathbb{R}_{\ge 0}$, the initial value problem associated to~\eqref{eq:general_system} has a unique (global) solution \cite[Theorem 8.1.5]{Arnold1974}.

 It is now useful to recall the definitions of stability, stabilisability and boundedness that we use in the paper. We say that an event $J\in\mathcal{A}$ happens \emph{almost surely} if $\mathcal{P}(J) = 1$, see \emph{e.g.} \cite{Oksendal2003}. We use the definition of almost sure asymptotic stability as given in \cite{Kozin1969} or as defined equivalently for linear systems in \cite[Section 11.4]{Arnold1974}, \cite[Section 2]{Kozin1963}.

\begin{definition}
	System~\eqref{eq:general_system} is almost surely asymptotically stabilisable if there exists a matrix $K$ such that \eqref{eq:general_system} with $u_t = Kx_t$ and $\omega \equiv 0$ is almost surely asymptotically stable.
\end{definition}

\begin{remark}
There are no available criteria for the choice of a feedback gain $K$ which achieves almost sure asymptotic stability of the closed-loop system in the general case that $A$ and $F$ do not commute (see \cite[Section 8.5]{Arnold1974} for the commuting case). However, since mean-square asymptotic stability implies almost sure asymptotic stability \cite[Chapter 11]{Arnold1974}, and the former can be easily enforced \cite[Section 11.3]{Arnold1974}, we can choose the matrix $K$ so that the stability conditions in the mean-square sense are met.
\end{remark}

\begin{definition}
    The stochastic process $a_t$, with $t\in \mathcal{T}\subseteq \mathbb{R}$, is almost surely bounded if there exists $0<M<\infty$ such that 
    $
        \mathcal{P}\left(\sup_{t\in \mathcal{T}}\Vert a_t \Vert < M\right) = 1.
    $
\end{definition}
\begin{remark}
        In the remainder, unless otherwise specified, we drop the wording ``almost sure(ly)'' whenever we refer to stability, stabilisability and boundedness, with the understanding that it remains implicit. 
\end{remark}
As the output regulation problem imposes requirements on the steady state of the closed-loop system, it is essential to provide a characterisation of the steady state of system~\eqref{eq:general_system}. The steady-state response of a stochastic system is the response of the system in the limit as the initial time tends to $-\infty$, or, equivalently, after an infinite amount of time has passed. Let $\Phi_t\in \mathbb{R}^{n\times n}$ be the fundamental matrix of the homogeneous equation corresponding to system~\eqref{eq:general_system}, \emph{i.e.}
    $d\Phi_t = (Adt + Fd\mathcal{W}_t)\Phi_t$.
Then, the following lemma (the proof of which can be found in \cite{Scarciotti2018bis}) characterises the steady-state response of a linear stochastic system of the form~\eqref{eq:general_system} when $u_t = 0$.

\begin{lemma}
\label{lemma:general_steady_state}
\cite{Scarciotti2018bis} Consider the interconnection of system~\eqref{eq:general_system} and the signal generator~\eqref{eq:exogenous_system} with $u_t=0$. Suppose that Assumption~\ref{assumption:ex_system_marg_stable} holds and that system~\eqref{eq:general_system} is asymptotically stable. Then the steady-state response of the output $y^c_t$ is
    $y^{c,ss}_t = C\Pi_t^{ss}\omega(t)$,
where $\Pi_t^{ss} \in \mathbb{R}^{n\times \nu}$, given by
$$
\Pi_t^{ss} \!=\! \Phi_t\!\left[ \int_{-\infty}^t \!\!\!\!\!\Phi_\tau^{-1}(P-FR)e^{S\tau}d\tau +\! \int_{-\infty}^t \!\!\!\!\!\Phi_\tau^{-1}Re^{S\tau}d\mathcal{W}_\tau\right]\!e^{-St}\!\!,
$$
is the steady-state response of
\begin{equation}
    d\Pi_t = (A\Pi_t - \Pi_tS + P)dt + (F\Pi_t + R)d\mathcal{W}_t.
\end{equation}
\end{lemma}

\section{Ideal Full-Information Problem}
\label{section:ideal_fi}
In this section we state the ideal full-information (\emph{i.e.} the state of the system and the Brownian motion are available for feedback) output regulation problem and we provide its solution. This result is instrumental for the development of a causal solution in the next sections.

\begin{problem}
\label{problem:IFI}
	(Ideal Full-Information Output Regulation \mbox{Problem}). Consider system~\eqref{eq:general_system}, driven by the signal generator~\eqref{eq:exogenous_system}. The \textit{ideal full-information output regulation problem} consists in determining a regulator such that the following conditions hold.
	\begin{description}
		\item[\textbf{(S$_I^F$)}] The closed-loop system obtained by interconnecting system~\eqref{eq:general_system} and the regulator with $\omega\equiv 0$ is asymptotically stable.
		\item[\textbf{(R$_I^F$)}] The closed-loop system obtained by interconnecting system~\eqref{eq:general_system}, the signal generator~\eqref{eq:exogenous_system} and the regulator satisfies $\lim_{t\rightarrow \infty} e_t = 0$ almost surely for any $(x_0, \omega_0)\in \mathbb{R}^n\times\mathbb{R}^\nu$. 
	\end{description}
\end{problem}

We look for a static regulator of the form
\begin{equation}
	\label{eq:static_regulator}
		u_t = Kx_t + \Gamma_t\omega,
	\end{equation}
where $K\in \mathbb{R}^{1\times n}$ and $\Gamma_t\in \mathbb{R}^{1\times\nu}$ is bounded. We first present a preliminary result. Namely, we assume that the condition \textbf{(S$_I^F$)} holds for some choice of the gain $K$ and we characterise the satisfaction of condition \textbf{(R$_I^F$)} via the solution of stochastic differential-algebraic equations.

\begin{lemma}
\label{lemma:preliminary_IFI}
    Consider Problem~\ref{problem:IFI} and let Assumption~\ref{assumption:ex_system_marg_stable} hold. Suppose there exists $K$ such that condition \textbf{(S$_I^F$)} holds and let $\Gamma_t\in \mathbb{R}^{1\times\nu}$ be bounded. Then condition \textbf{(R$_I^F$)} holds if and only if there exists a bounded matrix $\Pi_t\in \mathbb{R}^{n\times\nu}$ solving the equations
    \begin{equation}
	\label{eq:IFI_regulator_equations_with_K}
	\begin{aligned}
	d\Pi_t &= \left[(A+BK)\Pi_t - \Pi_t S + P + B\Gamma_t\right]dt\, +\\ &\quad \qquad\qquad \qquad \left[(F+GK)\Pi_t + R + G\Gamma_t\right]d\mathcal{W}_t,\\
	0 &= \lim_{t\rightarrow \infty}[(C+DK)\Pi_t + Q + D\Gamma_t]\quad \text{almost surely}.
	\end{aligned}
	\end{equation}
\end{lemma}

\begin{proof}
    By Assumption~\ref{assumption:ex_system_marg_stable}, the stability of the closed-loop system and Lemma~\ref{lemma:general_steady_state}, the matrix $\Pi_t$ solving~\eqref{eq:IFI_regulator_equations_with_K} exists. Let $u_t = Kx_t + \Gamma_t \omega$ and $\chi_t = x_t - \Pi_t \omega$. Then
        \begin{equation}
        \label{eq:closed_loop_auxiliary_state}
        \begin{aligned}
        d\chi_t &= (A+BK)\chi_tdt + (F+GK)\chi_td\mathcal{W}_t,\\
        e_t &= (C+DK)\chi_t + \left[(C+DK)\Pi_t + Q + D\Gamma_t\right]\omega.
        \end{aligned}
    \end{equation}
    Therefore, since \textbf{(S$_I^F$)} holds, $\lim_{t\rightarrow\infty} \chi_t = 0$, or, equivalently, $x_t$ converges to $\Pi_t\omega$, almost surely.
    
     \emph{Sufficiency:} assume that $\Pi_t$ satisfying~\eqref{eq:IFI_regulator_equations_with_K} is bounded. Then the steady-state response of the state of the closed-loop system is bounded. Moreover, since $\lim_{t\rightarrow\infty} \chi_t = 0$ and Assumption~\ref{assumption:ex_system_marg_stable} and the second equation in~\eqref{eq:IFI_regulator_equations_with_K} hold, then the tracking error of the closed-loop system, given by the second equation in~\eqref{eq:closed_loop_auxiliary_state}, satisfies $\lim_{t\rightarrow\infty} e_t = 0$ almost surely, \emph{i.e.} condition~\textbf{(R$_I^F$)} is satisfied.
    
    \emph{Necessity:} assume condition~\textbf{(R$_I^F$)} is satisfied. Since $\lim_{t\rightarrow\infty} \chi_t = 0$, then by the second equation in~\eqref{eq:closed_loop_auxiliary_state}
    \begin{equation}
        \lim_{t\rightarrow\infty}e_t = \lim_{t\rightarrow\infty}\left[(C+DK)\Pi_t + Q + D\Gamma_t\right]\omega.
    \end{equation}
    By assumption, $\lim_{t\rightarrow\infty}e_t = 0$ almost surely for all $\omega$, hence
    \begin{equation}
        0 = \lim_{t\rightarrow \infty}[(C+DK)\Pi_t + Q + D\Gamma_t] \quad \text{almost surely}.
    \end{equation}
    To show that $\Pi_t$ is bounded, observe that by Lemma~\ref{lemma:general_steady_state} the solution of the differential equation in~\eqref{eq:IFI_regulator_equations_with_K} is
    \begin{multline}
    \Pi_t = \widetilde{\Phi}_t\Pi_{t_0} e^{-St} +\\ \quad \quad \widetilde{\Phi}_t\left[ \int_{t_0}^t \widetilde{\Phi}_\tau^{-1}(P+B\Gamma_\tau-(F+GK)(R+G\Gamma_\tau))e^{S\tau}d\tau \right.\\ + 
            \left. \int_{t_0}^t \widetilde{\Phi}_\tau^{-1}(R+G\Gamma_\tau)e^{S\tau}d\mathcal{W}_\tau\right]e^{-St},
\end{multline}
where $\widetilde{\Phi}_t$ is the fundamental matrix associated to the autonomous system obtained from~\eqref{eq:general_system} with $u_t = Kx_t$ and $\omega \equiv 0$. Since \textbf{(S$_I^F$)} holds, $\widetilde{\Phi}_t$ converges exponentially to zero and, since the matrix $\Gamma_t$ is bounded, $\Pi_t$ is bounded as well.
\end{proof}

The following assumption is necessary for the problem to be solved.

\begin{assumption}
	\label{assumption:system_stabilisable}
	System~\eqref{eq:general_system} with $\omega \equiv 0$ is asymptotically stabilisable.
\end{assumption}

The next theorem provides the solution to the ideal full-information output regulation problem.

\begin{theorem}
	\label{theorem:IFI_solution}
	 Consider the ideal full-information regulator problem. Suppose Assumptions~\ref{assumption:ex_system_marg_stable} and \ref{assumption:system_stabilisable} hold. Then there exist matrices $K$ and $\Gamma_t$ such that the control law~\eqref{eq:static_regulator} is bounded and solves Problem~\ref{problem:IFI} if and only if there exist bounded matrices $\Pi_t\in\mathbb{R}^{n\times\nu}$ and $\Lambda_t\in\mathbb{R}^{1\times\nu}$ solving the equations
	\begin{equation}
	\label{eq:IFI_regulator_equations}
	\begin{aligned}
	d\Pi_t &= \left[A\Pi_t\! -\! \Pi_t S\! +\! P\! +\! B\Lambda_t\right]dt + \left[F\Pi_t\! +\! R\! +\! G\Lambda_t\right]d\mathcal{W}_t,\\
	0 &= \lim_{t\rightarrow \infty}[C\Pi_t + Q + D\Lambda_t]\quad \text{almost surely}.
	\end{aligned}
	\end{equation}
\end{theorem}
\begin{proof}
    \emph{Necessity:} assume there exist $K$ and $\Gamma_t$ such that the regulator~\eqref{eq:static_regulator} is bounded and satisfies conditions~\textbf{(S$_I^F$)} and \textbf{(R$_I^F$)}. Then by Lemma~\ref{lemma:preliminary_IFI} there exists a bounded $\Pi_t$ such that equations~\eqref{eq:IFI_regulator_equations_with_K} hold and $\Lambda_t = K\Pi_t + \Gamma_t$ is bounded and satisfies equations~\eqref{eq:IFI_regulator_equations}.
    
    \emph{Sufficiency:} assume bounded matrices $\Pi_t$ and $\Lambda_t$ solving~\eqref{eq:IFI_regulator_equations} exist. Then let $K$ be any matrix such that the closed-loop system is asymptotically stable (this is possible because Assumption~\ref{assumption:system_stabilisable} holds). Let $\Gamma_t = \Lambda_t - K\Pi_t$ and note that $\Gamma_t$ is bounded. With these selections of $K$ and $\Gamma_t$, condition~\textbf{(S$_I^F$)} holds. To show that condition~\textbf{(R$_I^F$)} holds it suffices to show that equations~\eqref{eq:IFI_regulator_equations_with_K} hold with these selections. Indeed, substituting $\Gamma_t$ in~\eqref{eq:IFI_regulator_equations_with_K}, equations~\eqref{eq:IFI_regulator_equations} are obtained, which hold by assumption. Therefore \textbf{(R$_I^F$)} holds.
\end{proof}

Theorem~\ref{theorem:IFI_solution} solves the ideal full-information problem. Indeed, if bounded $\Pi_t$ and $\Lambda_t$ solving~\eqref{eq:IFI_regulator_equations} are found, the control law
\begin{equation}
    u_t = Kx_t + (\Lambda_t - K\Pi_t)\omega
\end{equation}
solves Problem~\ref{problem:IFI} with any $K$ such that the closed-loop system is asymptotically stable.

\begin{remark}
The existence of bounded solutions $\Pi_t$ and $\Lambda_t$ solving the regulator equations~\eqref{eq:IFI_regulator_equations} can be characterised with the existence of steady-state solutions $\Pi_t^*$ and $\Lambda_t^*$ of
    \begin{equation}
	\label{eq:IFI_regulator_equations_ss}
	\begin{aligned}
	d\Pi_t^*\! &= \!\left[A\Pi_t^*\!\! -\! \Pi_t^* S\!\! +\!\! P\!\! +\!\! B\Lambda_t^*\right]dt\! +\! \left[F\Pi_t^*\! +\! R\! +\! G\Lambda_t^*\right]d\mathcal{W}_t,\\
	0 \!&=\! C\Pi_t^* + Q + D\Lambda_t^*\quad \text{almost surely}.
	\end{aligned}
	\end{equation}
   In fact, if such $\Pi_t$ and $\Lambda_t$ exist, for which the second regulator equation in \eqref{eq:IFI_regulator_equations} holds in the limit for $t$ tending to infinity, then there also exist two initial conditions $\Pi_{t_0}^*$ and $\Lambda_{t_0}^*$ such that the matrices $\Pi_t^*$ and $\Lambda_t^*$ satisfy the second regulator equation in \eqref{eq:IFI_regulator_equations_ss} identically for all $t> t_0$. This shows that, in the remainder of the paper, all the results  where regulator equations hold in the limit as $t$ tends to infinity can be equivalently replaced by analogous equations holding identically for all $t\ge 0$. We keep the limit formulation because it is more practical: it is easier to start from any initial condition and compute $\Pi_t$ and $\Lambda_t$ than to find $\Pi_{t_0}^*$ and $\Lambda_{t_0}^*$.
\end{remark}

\subsection{Solvability of the Regulator Equations}
We now discuss under which conditions the regulator equations~\eqref{eq:IFI_regulator_equations} are solvable. To this end, we make use of the definition of \emph{stochastic relative degree} given in \cite{Mellone2019ter}. Therein, a detailed definition of relative degree is introduced for a general class of nonlinear stochastic systems. For the purposes of the present paper, we specialise the definition to the case of linear stochastic systems of the form~\eqref{eq:general_system}. Thus, the stochastic relative degree of system~\eqref{eq:general_system} is either zero when $D\ne 0$ or the smallest $r$, with $1\le r \le n$, such that
\begin{itemize}
    \item $CA^{k}B = CA^{k}G = 0$, $CA^{k}F = 0$, $CA^{k}R = 0$ for all $k\in \{0,...,r-2\}$,
    \item $CA^{r-1}B \ne 0$ or  $CA^{r-1}G \ne 0$.
\end{itemize}

For simplicity, we first consider the case $D=0$. We assume that the relative degree of system~\eqref{eq:general_system} is $1\le r\le n$ and we set $b = CA^{r-1}B$ and $g = CA^{r-1}G$. By definition, either $b\ne 0$ or $g \ne 0$. To streamline the presentation, we set $Q_0 = Q$, $Q_i = CA^{i-1}P + Q_iS$, for $i= 1,..,r$ and $\zeta_t^{(i)} = CA^{i}x_t + Q_i\omega$ for $i= 0,..,r-1$ (note that $\zeta_t^{(0)} = e_t$).

We now provide the main result of this section, \emph{i.e.} we show that for a system with relative degree $r$ the solution of the regulator equations~\eqref{eq:IFI_regulator_equations} can be obtained by solving a system composed of a stochastic differential equation, a stochastic integral equation and $r$ algebraic equations.

\begin{lemma}
\label{lemma:regulator_equations_relative_degree}
Assume that \eqref{eq:general_system} has stochastic relative degree $r >0$. The matrices $\Pi_t$ and $\Lambda_t$ are solutions of the regulator equations~\eqref{eq:IFI_regulator_equations} if and only if they solve the equations
\begin{equation}
\label{eq:regulator_equations_differentiated}
\begin{aligned}
    d\Pi_t &= \left[A\Pi_t\! -\! \Pi_t S\! +\! P\! +\! B\Lambda_t\right]dt + \left[F\Pi_t\! +\! R\! +\! G\Lambda_t\right]d\mathcal{W}_t,\\
    0 &= \lim_{t\rightarrow\infty}[ (CA^r\Pi_t + b\Lambda_t + Q_{r})dt + \\ &\quad \, (CA^{r-1}F\Pi_t + g\Lambda_t + CA^{r-1}R)d\mathcal{W}_t]\quad \text{almost surely},\\
    0 &= \lim_{t\rightarrow\infty}[ CA^i\Pi_t + Q_i] \quad \text{almost surely}, \quad  i = 0,...,r-1.\\
\end{aligned}
\end{equation}
\end{lemma}
\begin{proof} \emph{Sufficiency:} The first equation in~\eqref{eq:regulator_equations_differentiated} is the same as in~\eqref{eq:IFI_regulator_equations}. Note that since~\eqref{eq:general_system} has stochastic relative degree $r>0$, $D=0$ and consequently the second equation in~\eqref{eq:IFI_regulator_equations} is equivalent to the third equation in~\eqref{eq:regulator_equations_differentiated} with $i = 0$. Therefore, if~\eqref{eq:regulator_equations_differentiated} hold then~\eqref{eq:IFI_regulator_equations} hold.\\
    \emph{Necessity:} By Theorem~\ref{theorem:IFI_solution}, if $\Pi_t$ and $\Lambda_t$ solve~\eqref{eq:IFI_regulator_equations}, then the steady-state tracking error of the closed-loop system is identically zero. Then necessarily $de_t = d\zeta_t^{(0)} = Cdx_t + QS\omega dt$ tends to zero. Replacing the expression of $dx_t$ with \eqref{eq:general_system} yields
    \begin{equation}
        \lim_{t\rightarrow\infty}d\zeta_t^{(0)}= \lim_{t\rightarrow\infty}[(CAx_t + Q_1\omega)dt]= \lim_{t\rightarrow\infty}[\zeta_t^{(1)}dt] = 0.
    \end{equation}
    At steady state this is verified for all $\omega$ if and only if $\lim_{t\rightarrow\infty}[CA\Pi_t + Q_1] = 0$. Repeating the same argument and using the definition of stochastic relative degree,
    \begin{equation}
        \lim_{t\rightarrow\infty}[d\zeta_t^{(i-1)}] = \lim_{t\rightarrow\infty}[\zeta_t^{(i)}dt] =\lim_{t\rightarrow\infty}[\!(CA^{i}x_t + Q_{i}\omega)dt],
    \end{equation}
    for all $i = 1,2,...,r-1$, which, at steady state, is satisfied for all $\omega$ if and only if $\lim_{t\rightarrow\infty}[CA^i\Pi_t + Q_i] = 0$ for all $i = 1,2,...,r-1$. Finally,
    \begin{multline}
        \lim_{t\rightarrow\infty}[d\zeta_t^{(r-1)}] = \lim_{t\rightarrow\infty}[\!(CA^{r}x_t + bu + Q_{r}\omega)dt +\\ (CA^{r-1}Fx_t + gu + CA^{r-1}R\omega)d\mathcal{W}_t] = 0
    \end{multline}
    is satisfied at steady state for all $\omega$ if and only if
    \begin{multline}
        0 = \lim_{t\rightarrow\infty}[ (CA^r\Pi_t + b\Lambda_t + Q_{r})dt \\+ (CA^{r-1}F\Pi_t + g\Lambda_t + CA^{r-1}R)d\mathcal{W}_t],
    \end{multline}
    which concludes the proof.
\end{proof}
Finding bounded solutions of equations~\eqref{eq:regulator_equations_differentiated}, as well as solvability conditions for said equations, is not, in general, a trivial task. However, when either $G = 0$ or $B=0$, it is possible to deduce solvability conditions for~\eqref{eq:regulator_equations_differentiated} and hence~\eqref{eq:IFI_regulator_equations}. We assume for the time being that $G=0$ (the case $B=0$ being analogous).
\begin{assumption}
\label{assumption:G_zero}
    The matrix $G$ is zero.
\end{assumption}
Let $\widehat{\Phi}_t$ be the fundamental matrix of the autonomous stochastic system
\begin{equation}
\label{eq:autonomous_pi_system}
    dx_t= A_\pi x_tdt + F_\pi x_td\mathcal{W}_t,
\end{equation}
where $A_\pi = A-BCA^rb^{-1}$ and $F_\pi = F - BCA^{r-1}Fb^{-1}$.
\begin{definition}
\label{definition:non-resonance}
    (Non-resonance Condition) Systems~\eqref{eq:general_system} and~\eqref{eq:exogenous_system} are \emph{non-resonant} if
    $
        \Upsilon_t := e^{-S^\top t}\otimes \widehat{\Phi}_t
    $
    converges exponentially to zero almost surely.
\end{definition}
We are now ready to provide the main result of this section, \emph{i.e.} the characterisation of the solvability of the regulator equations for arbitrary matrices $P$, $R$ and $Q$.

\begin{proposition}
\label{proposition:solvability_regulator_equations}
Under Assumption~\ref{assumption:G_zero}, there exist bounded solutions $\Pi_t$ and $\Lambda_t$ to the regulator equations~\eqref{eq:IFI_regulator_equations} for any $P$, $R$ and $Q$ if and only if systems~\eqref{eq:general_system} and \eqref{eq:exogenous_system} are non-resonant.
\end{proposition}

\begin{proof}
    \emph{Sufficiency:} assume that systems~\eqref{eq:general_system} and \eqref{eq:exogenous_system} are non-resonant. Then we have to show that it is possible to find bounded solutions to the equations~\eqref{eq:regulator_equations_differentiated}. To this end, consider the system
    \begin{equation}
    \label{eq:forced_pi_system}
        d\Pi_t = (A_\pi \Pi_t - \Pi_tS + P_{\pi,t})dt + (F_\pi\Pi_t + R_{\pi,t})d\mathcal{W}_t,
    \end{equation}
    with $P_{\pi} = P-BQ_rb^{-1}$, $R_{\pi} = R-BCA^{r-1}Rb^{-1}$. This is a stochastic differential equation in the variable~$\Pi_t$, which has solution
    \begin{multline}
    \label{eq:Pi_explicit}
    \Pi_t = \widehat{\Phi}_t\left[\Pi_{t_0} + \int_{t_0}^{t} \widehat{\Phi}^{-1}_\tau(P_{\pi,\tau} - F_\pi R_{\pi,\tau}) e^{S\tau}d\tau +\right. \\
    \left. \int_{t_0}^t\widehat{\Phi}_\tau^{-1} R_{\pi,\tau} e^{S\tau}d\mathcal{W}_\tau\right]e^{-St}.
    \end{multline}
    Applying the vectorisation operator and by the non-resonance condition,  $\lim_{t_0\rightarrow -\infty}\vect{\Pi_t}$ exists and is bounded for all $t\in \mathbb{R}$ almost surely for any $P$, $R$ and $Q$.  Now let $\Lambda_t$ be a matrix satisfying
    \begin{multline}
    \label{eq:Lambda_explicit}
        \Lambda_tdt := -(CA^r\Pi_t +  Q_r)b^{-1}dt -\\ (CA^{r-1}F\Pi_t + CA^{r-1}R)b^{-1}d\mathcal{W}_t,
    \end{multline}
    and observe that $\Lambda_t$ is bounded. Moreover, substituting this expression in~\eqref{eq:forced_pi_system} yields
	\begin{equation}
	\label{eq:regulator_equations_partial}
	\begin{aligned}
	d\Pi_t &= \left[A\Pi_t\! -\! \Pi_t S\! +\! P\! +\! B\Lambda_t\right]dt + \left[F\Pi_t\! +\! R\right]d\mathcal{W}_t,\\
	0 &= (CA^r\Pi_t + b\Lambda_t + Q_{r})dt + \\ &\quad \quad \quad\,(CA^{r-1}F\Pi_t + CA^{r-1}R)d\mathcal{W}_t\quad \text{almost surely}.
	\end{aligned}
	\end{equation}
	Finally, choosing the initial condition $\Pi_{t_0}$ such that
    \begin{equation}
        \begin{bmatrix}
        C\\ \vdots \\ CA^{r-1}
        \end{bmatrix}\Pi_{t_0} = -\begin{bmatrix}
        Q_0\\ \vdots \\ Q_{r-1}
        \end{bmatrix}
    \end{equation}
    is satisfied yields, together with the second equation in \eqref{eq:regulator_equations_partial}, the satisfaction of $CA^i\Pi_t + Q_i = 0$ for all $i=0,...,r-1$. Therefore $\Pi_t$ and $\Lambda_t$ are bounded and solve~\eqref{eq:regulator_equations_differentiated} (where $G=0$ by Assumption~\ref{assumption:G_zero}), hence \eqref{eq:IFI_regulator_equations} by Lemma~\ref{lemma:regulator_equations_relative_degree}.
	
    \emph{Necessity:} assume that bounded $\Pi_t$ and $\Lambda_t$ solving~\eqref{eq:IFI_regulator_equations} exist for any $P$, $R$ and $Q$. Then necessarily they solve equations~\eqref{eq:regulator_equations_differentiated}. Therefore let $\Sigma_t$ be a bounded matrix satisfying
    \begin{multline}
        \Sigma_tdt = (CA^r\Pi_t + b\Lambda_t + Q_{r})dt + CA^{r-1}(F\Pi_t +  R)d\mathcal{W}_t,
    \end{multline}
    hence
    \begin{equation}
        \Lambda_tdt\! =\! -[(CA^r\Pi_t - \Sigma_t + Q_{r})dt +CA^{r-1}(F\Pi_t + R)d\mathcal{W}_t]b^{-1}\!,
    \end{equation}
     Replacing this in the first equation of~\eqref{eq:regulator_equations_differentiated} yields
     \begin{equation}
        d\Pi_t = (A_\pi \Pi_t - \Pi_tS + P_{\pi,t})dt + (F_\pi\Pi_t + R_{\pi})d\mathcal{W}_t,
    \end{equation}
      with $P_{\pi,t} = P-B(Q_r - \Sigma_t )b^{-1}$. The solution $\Pi_t$ of the previous differential equation is bounded by assumption. But this holds only if $\Upsilon_t$ converges exponentially to zero almost surely, \emph{i.e.} systems~\eqref{eq:general_system} and~\eqref{eq:exogenous_system} are non-resonant.
\end{proof}

\begin{remark}
In the case of zero relative degree, \emph{i.e.} $D\ne 0$, the non-resonance condition remains formally unchanged. The assumption that either $B$ or $G$ are zero is not necessary and $A_\pi = A-BCD^{-1}$ and $F_\pi = F-GCD^{-1}$.  
\end{remark}

In the light of the previous result, the following corollary states a sufficient condition for the solvability of the ideal full-information problem.

\begin{assumption}
    \label{assumption:Pi_system_stable}
    System~~\eqref{eq:autonomous_pi_system} is asymptotically stable.
\end{assumption}

\begin{corollary}
    Under Assumptions~\ref{assumption:ex_system_marg_stable}, \ref{assumption:system_stabilisable}, \ref{assumption:G_zero} and \ref{assumption:Pi_system_stable}, Problem~\ref{problem:IFI} is solvable by the control law~\eqref{eq:static_regulator}.
\end{corollary}
\begin{proof}
    This is trivial since Assumptions~\ref{assumption:ex_system_marg_stable} and~\ref{assumption:Pi_system_stable} together imply the non-resonance condition and, therefore, the existence of bounded solutions to~\eqref{eq:IFI_regulator_equations}.
\end{proof}

In the rest of the paper, the existence of a solution to the output regulation problem is guaranteed by assuming the following.
\begin{assumption}
\label{assumption:nonresonance}
    Systems~\eqref{eq:general_system} and~\eqref{eq:exogenous_system} are non-resonant.
\end{assumption}

\begin{example}
\label{example:scalar_solvability}
To illustrate the validity of the solvability condition, consider the following scalar example: $A = 0.2$, $B = 0.5$, $F = 0.3$, $G = 0.2$, $C= c\in \mathbb{R}$, $D = 0.1$, $S = 0$, $P\in \mathbb{R}$, $R\in \mathbb{R}$ and $Q\in \mathbb{R}$. This yields  $A_\pi = 0.2-5c$ and $F_\pi = 0.3-2c$. Almost sure asymptotic stability for this scalar system is obtained if $2A_\pi - F_\pi^2 = -4c^2 - 8.8c + 0.31 <0$, \emph{i.e.} $c\in (-\infty, -2.23)\cup (0.034, +\infty)$.  Moreover, note that asymptotic mean-square stability, which implies almost sure asymptotic stability (see \cite[Section 11.4]{Arnold1974}), is achieved if $2A_\pi + F_\pi^2 = 4c^2 - 11.2c +0.49 < 0$, \emph{i.e.} for $c\in (0.044,2.75)$.
\begin{figure}
    \centering
    \includegraphics[width=0.9\columnwidth]{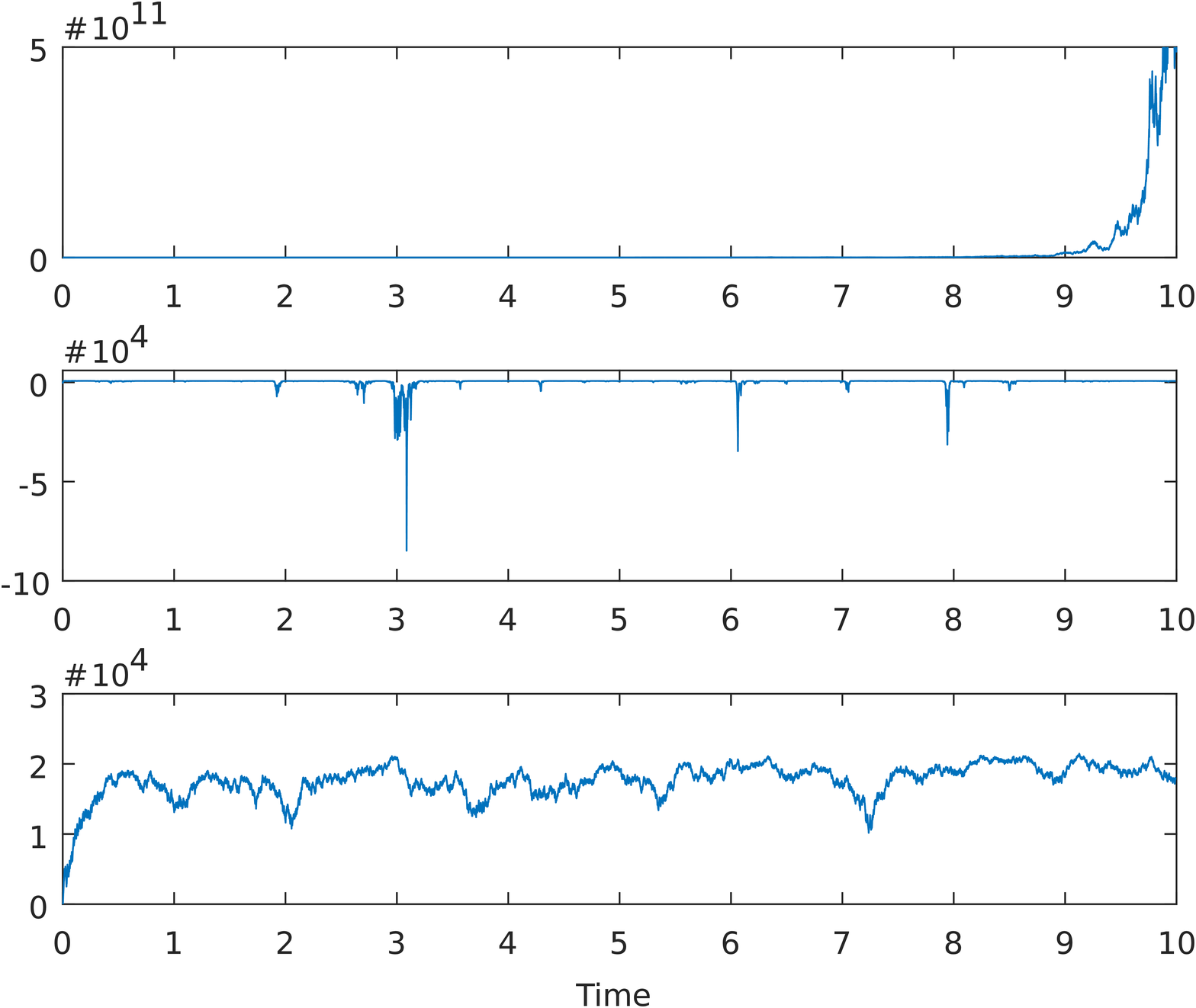}
    \caption{Time history of $\Pi_t$ in Example~\ref{example:scalar_solvability} for $c = -0.5$ (top plot), $c = -5$ (middle plot) and $c=0.5$ (bottom plot).}
    \label{fig:Pi_histories_solvability_example}
\end{figure}
Figure~\ref{fig:Pi_histories_solvability_example} shows the time histories of $\Pi_t$ when $c = -0.5$ (top plot), $c = -5$ (middle plot) or $c = 0.5$ (bottom plot). The first choice is such that $\Pi_t$ has unstable dynamics, so it diverges (top plot). The second and third choices are such that the system is asymptotically stable almost surely (middle plot) and in mean square (bottom plot), respectively; $\Pi_t$ remains bounded in both plots but the bottom one displays a more regular evolution, suggesting the boundedness of the second moment of $\Pi_t$. The plots show that a stronger concept of stability, such as mean-square asymptotic stability, would be unnecessarily conservative for the solvability of the problem. $\hfill\qedsymbol$
\end{example}
We now want to show that the non-resonance condition given in Definition~\ref{definition:non-resonance} is the natural extension of the deterministic solvability condition of the output regulation problem. In the deterministic case $\widehat{\Phi}_t = e^{A_\pi t}$ and $\Upsilon_t = e^{-S^\top t}\otimes e^{A_\pi t}=e^{(I\otimes A_\pi - S^\top\otimes I)t}$, which converges exponentially to zero if and only if
    \begin{equation}
    \label{eq:stability_Pi_deterministic}
        \sigma(I\otimes A_\pi - S^\top \otimes I)\subset \mathbb{C}_{<0}.
    \end{equation} In particular, this implies (but is not implied by)
    \begin{equation}
        \rank{I\otimes (A-BCA^rb^{-1}) - S^\top \otimes I}= n\nu
    \end{equation} and, therefore,
    \begin{equation}
        \rank{I\otimes b} \!+\! \rank{I\otimes (A-BCA^rb^{-1})\! -\! S^\top \otimes I} = (n+1)\nu.
    \end{equation}
    Using the properties of the Kronecker product and of the Schur complement, this is equivalent to
    \begin{equation}
        \text{rank}\begin{bmatrix}
        I\otimes A - S^\top\otimes I & I\otimes B\\ I\otimes CA^r & I\otimes b
        \end{bmatrix} = (n+1)\nu,
    \end{equation}
    which is in turn equivalent to \cite[Theorem 1.9]{Huang2004}
    \begin{equation}
    \label{eq:pre_deterministic_solvability}
        \text{rank}\begin{bmatrix}
        A - \lambda I & B \\ CA^r & b
        \end{bmatrix} = n+1, \quad \forall \lambda \in \sigma(S).
    \end{equation}
    Finally, it is possible to prove that~\eqref{eq:pre_deterministic_solvability} implies
    \begin{equation}
    \label{eq:deterministic_solvability}
        \text{rank}\begin{bmatrix}
        A - \lambda I & B \\ C & 0
        \end{bmatrix} = n+1, \quad \forall \lambda \in \sigma(S).
    \end{equation}
    This is the deterministic non-resonance condition, see \cite[Assumption 1.4]{Huang2004}. Note that Proposition~\ref{proposition:solvability_regulator_equations}, applied to deterministic systems, would state that the matrices $\Pi(t)$ and $\Lambda(t)$ solving the differential-algebraic equations
    \begin{equation}
    \label{eq:diff_regulator_equations_deterministic}
        \begin{aligned}
        \dot{\Pi}(t) &= A\Pi(t) - \Pi(t) S + P + B\Lambda(t),\\
        0 &= \lim_{t\rightarrow\infty}\left[C\Pi(t) + Q + D\Lambda(t)\right]
        \end{aligned}
    \end{equation}
    are bounded if and only if~\eqref{eq:stability_Pi_deterministic} holds. But we know that the solution of the problem in the deterministic case is given by the equilibrium of~\eqref{eq:diff_regulator_equations_deterministic}, the existence of which is equivalent to~\eqref{eq:deterministic_solvability}. However, \eqref{eq:stability_Pi_deterministic} is a stronger condition, sufficient but not necessary for~\eqref{eq:deterministic_solvability} to hold. This discrepancy is due to the fact that~\eqref{eq:stability_Pi_deterministic} is a condition on the stability of system~\eqref{eq:diff_regulator_equations_deterministic}, which is also necessary if we intend to solve the output regulation problem by integration of~\eqref{eq:diff_regulator_equations_deterministic}. Thus the exponential decay of $\Upsilon_t$ (hence condition~\eqref{eq:stability_Pi_deterministic}) is sufficient but not necessary in the deterministic framework. We now show, in contrast, that in the stochastic case it is also necessary. The reason for this is that there exist no constant matrices $\Pi$ and $\Lambda$ that are an equilibrium of the regulator equations~\eqref{eq:IFI_regulator_equations}. To see this, note that if there existed constant $\Pi$ and $\Lambda$ being an equilbrium of~\eqref{eq:IFI_regulator_equations}, then $0\! =\! A\Pi - \Pi S + P + B\Lambda$, $0\! =\! F\Pi + R + G\Lambda$ and $0 \!=\! C\Pi + Q + D\Lambda$ would hold, which is a system of linear equations with more constraints than unknowns. Therefore, the matrices $\Pi_t$ and $\Lambda_t$ can be obtained for any $P$, $R$, $Q$ only by integration of the regulator equations~\eqref{eq:IFI_regulator_equations}, which have a bounded solution if and only if $\Upsilon_t$ decays exponentially to zero almost surely.
    
    \subsection{A note on the internal model property}
    \label{section:internal_model}
    In this section we show that, if the regulator equations~\eqref{eq:IFI_regulator_equations} are solvable, it is possible to design a regulator that achieves a stochastic internal model property without using a feedback of the variable $\omega$. The development and terminology of this section follows the theory of the internal model regulator presented in \cite{Marconi2013}. We first describe the structure of the regulator and then comment on the stochastic internal model property.
    
    Consider the closed-loop system composed of~\eqref{eq:general_system}, \eqref{eq:exogenous_system} and
    \begin{equation}
    \label{eq:internal_model_regulator_general}
        dz_t = (\mathcal{G}_{im}^{1}z_t + \mathcal{G}_{im}^2e_t)dt, \qquad u_t = Kx_t  +K^z_t z_t,
    \end{equation}
    where $z_t \in \mathbb{R}^\nu$, $K_t^z \in \mathbb{R}^{1\times \nu}$, $\mathcal{G}_{im}^{1} \in \mathbb{R}^{\nu \times \nu}$ and $\mathcal{G}_{im}^2 \in \mathbb{R}^{\nu \times 1}$ are such that $(\mathcal{G}_{im}^{1},\mathcal{G}_{im}^2)$ is a controllable pair. Let $\bar{x}_t = \begin{bmatrix}
    x_t^\top & z_t^\top
    \end{bmatrix}^\top$. Then the closed-loop system is given by
    \begin{equation}
    \label{eq:imp_closed_loop}
        d\bar{x}_t = (\bar{A}\bar{x}_t + \bar{P}\omega)dt + (\bar F \bar{x}_t + \bar R \omega)d\mathcal{W}_t,
    \end{equation}
    with
    \begin{equation}
        \bar{A} = \begin{bmatrix}
        A + BK & BK^z_t\\ \mathcal{G}_{im}^2C & \mathcal{G}_{im}^1 + \mathcal{G}_{im}^2DK^z_t
        \end{bmatrix}, \quad \bar P = \begin{bmatrix}
        P \\ \mathcal{G}_{im}^2Q
        \end{bmatrix},
    \end{equation}
    \begin{equation}
        \bar{F} = \begin{bmatrix}
        F + GK & GK^z_t\\ 0 & 0
        \end{bmatrix}, \quad \bar P = \begin{bmatrix}
        R \\ 0
        \end{bmatrix}.
    \end{equation}
    
    The following result characterises the solution of Problem~\ref{problem:IFI} through a regulator of the form~\eqref{eq:internal_model_regulator_general}.
    
    \begin{lemma}
    \label{lemma:preliminary_imp}
    Consider Problem~\ref{problem:IFI} and let Assumption~\ref{assumption:ex_system_marg_stable} hold. Suppose there exists a regulator~\eqref{eq:internal_model_regulator_general} such that condition \textbf{(S$_I^F$)} holds. Then condition \textbf{(R$_I^F$)} holds if and only if there exist bounded matrices $\Pi_t\in \mathbb{R}^{n\times\nu}$, $\Lambda_t\in \mathbb{R}^{1\times\nu}$ and $\Pi^z_t \in \mathbb{R}^{\nu\times\nu}$ solving the equations
    \begin{align}
        d\Pi_t &= \left[A\Pi_t\! -\! \Pi_t S\! +\! P\! +\! B\Lambda_t\right]dt + \left[F\Pi_t\! +\! R\! +\! G\Lambda_t\right]d\mathcal{W}_t,\\
	    0 &= \lim_{t\rightarrow \infty}[C\Pi_t + Q + D\Lambda_t]\quad \text{almost surely}, \\
	    d\Pi_t^z &= \left[\mathcal{G}^1_{im}\Pi_t^z - \Pi_t^z S\right]dt,\\
	    \Lambda_t &= K\Pi_t + K^z_t \Pi^z_t.
	    \label{eq:internal_model_regulator_equations}
    \end{align}
    \end{lemma}
    \begin{proof}
        Let $\overline{\Pi}_t = \begin{bmatrix} \Pi_t & \Pi_t^z \end{bmatrix}$. Then note that equations~\eqref{eq:internal_model_regulator_equations} can be rewritten compactly as
        \begin{align}
           d\overline{\Pi}_t &= \left[ \bar A \overline{\Pi}_t - \overline{\Pi}_t S + \bar{P}\right]dt + [\bar F \overline{\Pi}_t + \bar R]d\mathcal{W}_t,\\
           0 &= \lim_{t\rightarrow \infty} \left[ \bar{C}\overline{\Pi}_t + Q\right] \qquad \text{almost surely,}
           \label{eq:closed_loop_imp_regulator_equations}
        \end{align}
        with $\bar C = \begin{bmatrix}
        C + DK & DK_t^z
        \end{bmatrix}$. Therefore, repeating the same arguments of Lemma~\ref{lemma:preliminary_IFI} for the closed-loop system~\eqref{eq:imp_closed_loop} with equations~\eqref{eq:closed_loop_imp_regulator_equations} playing the role of equations~\eqref{eq:IFI_regulator_equations_with_K},  $\bar{x}_t$ converges to $\overline{\Pi}_t\omega$ and the claim follows.
    \end{proof}
    A regulator such that equations~\eqref{eq:internal_model_regulator_equations} are satisfied is said to achieve the stochastic internal model property. We now discuss why this is the case. As shown previously in this section, the existence of bounded $\Pi_t$ and $\Lambda_t$ solving the regulator equations~\eqref{eq:IFI_regulator_equations} (equivalently, the first two equations in~\eqref{eq:internal_model_regulator_equations}) is equivalent to the existence of a steady state of the system such that $x_t^{ss} = \Pi_t^{ss}\omega$, $u_t^{ss} = \Lambda_t^{ss}\omega$ and $e_t^{ss} = 0$. Therefore, any regulator solving the full-information output regulation problem must be such that it generates all the signals generated by the output $y^{g}_t$ of the system
    \begin{equation}
        \dot{\omega} = S\omega, \qquad y^g_t = \Lambda_t \omega,
    \end{equation}
    when the tracking error is identically zero. This is indeed referred to as internal model property. Note that the regulator~\eqref{eq:internal_model_regulator_general} achieves this property if there exists a bounded solution $\Pi_t^z$ of the equations~\eqref{eq:internal_model_regulator_equations}. In fact, if this holds, the steady-state response $z^{ss}_t$ of the process $z_t$ is $\Pi_t^{z,ss}\omega$ while $e_t^{ss} = 0$ and the output of the regulator satisfies $u_t^{ss} = Kx_t^{ss} + K_t^z z_t^{ss} = (K\Pi_t^{ss} + K_t^z \Pi_t^{z,ss})\omega = \Lambda_t^{ss}\omega$ by the last equation in~\eqref{eq:internal_model_regulator_equations}, \emph{i.e.} the regulator~\eqref{eq:internal_model_regulator_general} possesses the stochastic internal model property. 
    
    We now look at how to design the regulator~\eqref{eq:internal_model_regulator_general} so that it achieves said property. It turns out that if a technical assumption is satisfied, then such a regulator can always be designed. In particular, we set\footnote{Although $\mathcal{G}_{im}^1$ is time-varying, we omit the subscript $t$ to avoid notation overload.} $\mathcal{G}_{im}^1 = H_{im} + L_{im}K_t^z$, with $H_{im}\in\mathbb{R}^{\nu \times \nu}$ and $L_{im}\in \mathbb{R}^{\nu \times 1}$ to be designed. Once the solutions $\Pi_t$ and $\Lambda_t$ of the first two regulator equations~\eqref{eq:internal_model_regulator_equations} have been found, the last two equations in~\eqref{eq:internal_model_regulator_equations} become
    \begin{align}
        d\Pi_t^z &= \left[H_{im}\Pi_t^z - \Pi_t^z S + L_{im}(\Lambda_t - K\Pi_t)\right]dt,\\
	    \Lambda_t &= K\Pi_t + K^z_t \Pi^z_t.
	    \label{eq:imp_regulator_equations}
    \end{align}
    
    \begin{proposition}
    \label{proposition:internal_model_regulator}
    Consider Problem~\ref{problem:IFI} and let Assumptions~\ref{assumption:ex_system_marg_stable} and \ref{assumption:system_stabilisable} hold. Assume that there exist bounded solutions $\Pi_t$ and $\Lambda_t$ of the regulator equations~\eqref{eq:IFI_regulator_equations}. Let the regulator~\eqref{eq:internal_model_regulator_general}, with the selection $\mathcal{G}_{im}^1 = H_{im} + L_{im}K_t^z$, be such that the couple $(H_{im}, L_{im})$ is controllable, $\sigma(H_{im}) \cap \sigma(S) = \emptyset$ and condition~\textbf{(S$_I^F$)} holds. If the solution $\Pi^z_t$ of the first equation in~\eqref{eq:imp_regulator_equations} is invertible almost surely for all $t\ge t_0$, then there exists a bounded $K^z_t = (\Lambda_t - K\Pi_t){(\Pi_t^z)}^{-1}$ such that $\Lambda_t = K\Pi_t + K^z_t \Pi^z_t$. Thus the regulator~\eqref{eq:internal_model_regulator_general} solves Problem~\ref{problem:IFI}.
    \end{proposition}
    \begin{proof}
        By Assumptions~\ref{assumption:ex_system_marg_stable} and \ref{assumption:system_stabilisable} it is possible to design the regulator~\eqref{eq:internal_model_regulator_general} as in the statement of this proposition, in particular achieving closed-loop stability. Moreover, as the eigenvalues of $H_{im}$ and $S$ are disjoint and the matrices $\Pi_t$ and $\Lambda_t$ are bounded, by Lemma~\ref{lemma:general_steady_state} there exists a bounded solution $\Pi_t^z$ of the first equation in~\eqref{eq:imp_regulator_equations}, therefore of the third equation in~\eqref{eq:internal_model_regulator_equations}. Additionally, if $\Pi_t^z$ is invertible almost surely, the selection $K^z_t = (\Lambda_t - K\Pi_t){(\Pi_t^z)}^{-1}$ makes the last equation in~\eqref{eq:internal_model_regulator_equations} hold. Therefore, by Lemma~\ref{lemma:preliminary_imp}, condition~\textbf{(R$_I^F$)} is satisfied and the regulator solves Problem~\ref{problem:IFI}.
    \end{proof}
   
   We now present a series of comments about the internal model regulator.
   
   \begin{remark}
   The requirement of the invertibility of the matrix $\Pi_t^z$ in Proposition~\ref{proposition:internal_model_regulator} is due to the fact that the regulator has the same dimension of the exosystem, thus the steady-state matrix $\Pi_t^z$ is square. More generally, for regulators of higher order, the requirement would be that the steady-state matrix is full rank at all times (see, \emph{e.g.}, \cite{Marconi2013}).
   \end{remark}
   
   \begin{remark}
   The use of a dynamic regulator of the form~\eqref{eq:internal_model_regulator_general} achieving the stochastic internal model property shows that it is possible to design control laws without using the exosystem state variable $\omega$ in the feedback loop, as it happens in control laws of the form~\eqref{eq:static_regulator}. Nevertheless, having shown that such internal model regulators can be designed using standard linear control techniques, in this paper we prefer to use controllers that adopt the variable $\omega$, because this helps us to keep the notation simpler in the following sections, with the understanding that the design of internal model regulators can be done trivially using the results just introduced.
   \end{remark}
   
   \begin{remark}
   It is evident that both the controls~\eqref{eq:static_regulator} and~\eqref{eq:internal_model_regulator_general} are not robust with respect to parametric uncertainties potentially affecting the system matrices. This is due to the fact that said controllers rely on the solution of the regulator equations, where these system matrices explicitly appears. In \cite{Marconi2013} this problem was overcome with an additional technical assumption on the nature of the uncertainty, which however does not hold in the stochastic case. Note that, although the controllers presented in the present paper are inherently robust with respect to stochastic uncertainties, the problem of designing a controller which achieves robust output regulation with respect to parametric uncertainties is still open at this stage.
   \end{remark}
   
\section{Approximate Full-Information Problem}
\label{section:approximate_fi}
The solution provided by Theorem~\ref{theorem:IFI_solution} is not implementable in real contexts. In fact, the integration of the regulator equations~\eqref{eq:IFI_regulator_equations} makes it necessary to access the signal $\mathcal{W}_t$. This quantity is never available for measure, which makes this approach impossible in practice. 
To deal with this fundamental issue, we introduce a hybrid architecture that periodically performs an \emph{a-posteriori} partial estimation of the noise that affected the system between sampling times. We then show that the solution of Problem~\ref{problem:IFI} can be approximated with a degree of accuracy that depends on the frequency of the sampling times. In order to accomplish this, we formulate an approximate problem and we show that it can be solved in practice (\emph{i.e.} access to the Brownian motion is not required) with a hybrid scheme.
\begin{problem}
	\label{problem:epsAFI}
	($\varepsilon$-Approximate Full-Information Output Regulation Problem). Consider system~\eqref{eq:general_system}, driven by the signal generator~\eqref{eq:exogenous_system}. The $\varepsilon$-\textit{approximate full-information output regulation problem} consists in determining a regulator such that the following conditions hold.
	\begin{description}
		\item[\textbf{(S$_A^F$)}] The closed-loop system obtained by interconnecting system~\eqref{eq:general_system} and the regulator with $\omega\equiv 0$ is asymptotically stable.
		\item[\textbf{(R$_A^F$)}] The closed-loop system obtained by interconnecting system~\eqref{eq:general_system}, the signal generator~\eqref{eq:exogenous_system} and the regulator yields a steady-state response of the tracking error $e^{ss}_t(\omega_0,\varepsilon)$, with $\varepsilon \in \mathbb{R}_{>0}$, which is bounded and such that
	\begin{equation}
			\begin{aligned}
			\lim_{\Vert \omega_0\Vert \rightarrow 0}e_t^{ss}(\omega_0, \varepsilon) &= 0,  \qquad \forall \varepsilon \in \mathbb{R}_{>0},\\
			\lim_{\varepsilon\rightarrow 0} e_t^{ss}(\omega_0, \varepsilon) &= 0,  \qquad\forall \omega_0\in \mathbb{R}^\nu,
			\end{aligned}
		\end{equation} 
		almost surely for any $x_0\in \mathbb{R}^n$.
	\end{description}
\end{problem}
By looking at the solution of the ideal problem, we now seek a static regulator of the form
\begin{equation}
	\label{eq:static_regulator_approx}
	u_t = Kx_t + \F{\widehat{\Gamma}_t}\omega,
	\end{equation}
where $K\in \mathbb{R}^{1\times n}$ and $\F{\widehat{\Gamma}_t}\in \mathbb{R}^{1\times\nu}$ is bounded. Observe that condition~\textbf{(S$_A^F$)} is equivalent to~\textbf{(S$_I^F$)} as the stabilisation of the system is independent of the Brownian motion in the case of full information. An element of novelty is introduced with the regulation condition~\textbf{(R$_A^F$)}. The rationale of this condition is that, although we allow for a tracking error which is almost surely non-zero, we recover the ideal case if the parameter $\varepsilon$ is chosen arbitrarily small. In Section~\ref{subsec:solution_AFI} it is shown that this parameter is the period at which samples of the state of the system are compared in order to obtain an estimate of the variations of the Brownian motion. Moreover, given the marginal stability of the exogenous system (see Assumption~\ref{assumption:ex_system_marg_stable}), for small $\Vert \omega_0 \Vert$ we also have small $\Vert \omega(t) \Vert$ for all $t \in \mathbb{R}_{>0}$.

\subsection{Reconstruction of the Brownian Motion}
\label{subsec:reconstruction_brownian_motion_FI}
We now review It\^{o}'s interpretation of a stochastic integral. This is instrumental in developing the theory of approximation of the Brownian motion.

Given a map $f:\mathbb{R}_{\ge 0}\times \nabla \rightarrow \mathbb{R}^n$, a stochastic integral is a stochastic process of the form
 \begin{equation}
 \label{eq:stochastic_integral}
 	X(t_n,w) = \int_{t_0}^{t_n} f(\tau,w)d\mathcal{W}_\tau
 \end{equation}    
which can be approximated as \cite[Section 2.1]{Gard1988}
\begin{equation}
\label{eq:stochastic_integral_approx}
	X(t_n,w) \approx \sum_{k=1}^n f(\tau_{k-1},w)(\mathcal{W}_{t_k} - \mathcal{W}_{t_{k-1}}).
\end{equation}
The value $\tau_{k-1}$ can be arbitrarily chosen in the interval $[t_{k-1}, t_k]$, thus yielding different intepretations and properties. It\^{o}'s interpretation of the stochastic integral~\eqref{eq:stochastic_integral}, which makes it causal, is given by taking the limit of~\eqref{eq:stochastic_integral_approx} with $\tau_{k-1} = t_{k-1}$ for $t_k - t_{k-1}$ tending to zero. As already stated, all the integrals are meant in It\^{o}'s sense in this paper.

In the remainder, it is assumed that the sequence $\{t_k\}_{k\ge 0}$ is defined such that $t_k - t_{k-1} = \varepsilon$ for all $k \in \mathbb{Z}_{>0}$ and we define $\Delta x(k) = x_{t_k} - x_{t_{k-1}}$. We now show that, if the forward-Euler scheme is adopted as a starting point for an approximation of the system dynamics, it is possible to compute an estimate $\Delta \widehat{W}_\varepsilon(k)$ of the variation of the Brownian motion $\Delta W_\varepsilon(k) = \mathcal{W}_{t_k} - \mathcal{W}_{t_{k-1}}$ that ``converges" to the stochastic differential $d\mathcal{W}_t$ as the parameter $\varepsilon$ converges to zero. Specifically, let $\mathcal{L}_I$ be the space of functions that are integrable in It\^o's sense. Then with the notation
\begin{equation}
    \Delta\widehat{W}_\varepsilon \xrightarrow{\varepsilon} d\mathcal{W}_t
\end{equation}
we mean that for all $f\in\mathcal{L}_I$
\begin{equation}
    \lim_{\varepsilon\rightarrow 0}\sum_{k}f(t_{k-1},w)\Delta \widehat{W}_\varepsilon(k) = \int_0^t f(\tau,w)d\mathcal{W}_\tau.
\end{equation}
Before proceeding, for ease of notation define the vector
\begin{equation}
	v(k) = Fx_{t_{k}} + Gu_{t_{k}} + R\omega(t_{k}), \quad \forall k\in \mathbb{Z}_{\ge 0}.
\end{equation}
which is the diffusion coefficient of the system dynamics evaluated at time $t_{k}$. The following assumption ensures that the noise persistently excites the system. 
\begin{assumption}
\label{assumption:v_not_zero_FI}
There exists $\delta \in \mathbb{R}_{>0}$ such that $\vert v(k) \vert > \delta$ almost surely for all $k\in \mathbb{Z}_{\ge 0}$.
\end{assumption}
We postpone the discussion of the rationale of this assumption to the end of this section. Under Assumption~\ref{assumption:v_not_zero_FI}, we define the Moore-Penrose left pseudo-inverse of $v(k)$ as $v(k)^+ = (v(k)^\top v(k))^{-1}v(k)^\top$, the norm of which is uniformly bounded.

\begin{lemma}
	\label{lemma:brownian_motion_estimation_FI}
	Consider system~\eqref{eq:general_system} and let Assumption~\ref{assumption:v_not_zero_FI} hold. Let the sequence $\{\Delta \F{\widehat{W}_\varepsilon}(k)\}_{k > 0}$ be defined as
	\begin{multline}
	\label{eq:dWest_definition_FI}
		\Delta \F{\widehat{W}_\varepsilon}(k) = v(k-1)^+[\Delta x(k) -\\ (Ax_{t_{k-1}} + Bu_{t_{k-1}} + P\omega(t_{k-1}))\varepsilon].
	\end{multline}
	Then  $\Delta \F{\widehat{W}_\varepsilon} \xrightarrow{\varepsilon} d\mathcal{W}_t$ almost surely.
\end{lemma}
\begin{proof}
	Let $k\in\mathbb{Z}_{>0}$. By \cite[Theorem 7.1]{Gard1988}
	\begin{multline}
		\Delta x(k) =  [Ax_{t_{k-1}} + Bu_{t_{k-1}} + P\omega(t_{k-1})]\varepsilon +\\ v(k-1)\Delta W_\varepsilon(k) + o(\varepsilon^2),
	\end{multline}
	holds, where $o(\varepsilon^2)$, which is the \textit{one-step truncation error} of the forward-Euler scheme, is an infinitesimal of the same order of $\varepsilon^2$. The previous expression can be rewritten as
	\begin{multline}
		v(k-1)\Delta W_\varepsilon(k) = \Delta x(k) -\\ [Ax_{t_{k-1}} + Bu_{t_{k-1}} + P\omega(t_{k-1})]\varepsilon + o(\varepsilon^2).
	\end{multline}
	Since $v(k-1)$ has full column rank almost surely, the expression
	\begin{multline}
		\Delta W_\varepsilon(k) = v(k-1)^+[\Delta x(k) -\\ (Ax_{t_{k-1}} + Bu_{t_{k-1}} + P\omega(t_{k-1}))\varepsilon + o(\varepsilon^2)],
	\end{multline}
	holds almost surely. Defining $\Delta \F{\widehat{W}_\varepsilon}(k)$ as in~\eqref{eq:dWest_definition_FI} yields
	\begin{equation}
	\label{eq:relation_DeltaW_true_est_FI}
		\Delta \F{\widehat{W}_\varepsilon}(k) = \Delta W_\varepsilon(k) + v(k-1)^+o(\varepsilon^2).
	\end{equation}
	almost surely. Let $f_t\in \mathcal{L}_I$. Then
	\begin{equation}
	    \sum_k f_{t_{k-1}} \Delta \F{\widehat{W}_\varepsilon}(k) =  \sum_k f_{t_{k-1}} (\Delta W_\varepsilon(k) + v(k-1)^+ o(\varepsilon^2)).
	\end{equation}
	Taking the limit of both sides for $\varepsilon$ that tends to zero yields $\Delta \widehat{W}_\varepsilon \xrightarrow{\varepsilon}d\mathcal{W}_t$, since for all $f\in\mathcal{L}_I$
	\begin{equation}
	    \lim_{\varepsilon \rightarrow 0} \sum_k f_{t_{k-1}}v(k-1)^+ o(\varepsilon^2) = 0 \quad \text{almost surely.}
	\end{equation}
\end{proof}

\subsection{Solution to the Approximate Problem}
\label{subsec:solution_AFI}
In this section we show how to employ the approximation of the variations of the Brownian motion in order to solve the $\varepsilon$-approximate full-information problem. To this end, we first give a preliminary result.

\begin{lemma}
	\label{lemma:steady_state_approximated}
	Consider the closed-loop system obtained interconnecting \eqref{eq:general_system}, \eqref{eq:exogenous_system} and \eqref{eq:static_regulator_approx} and let Assumptions~\ref{assumption:ex_system_marg_stable}, \ref{assumption:nonresonance} and \ref{assumption:v_not_zero_FI} hold. Then there exist bounded matrices $\F{\widehat{\Pi}_t}\in \mathbb{R}^{n\times\nu}$ and $\F{\widehat{\Lambda}_t}\in \mathbb{R}^{1\times \nu}$ solving
	\begin{equation}
	\label{eq:steady_state_regulation_FI}
	\begin{aligned}
		d \F{\widehat{\Pi}_t} &= [A\F{\widehat{\Pi}_t} - \F{\widehat{\Pi}_t} S + P + B\F{\widehat{\Lambda}_t}]dt,\\
		\F{\widehat{\Pi}_{t_k^+}} &= \F{\widehat{\Pi}_{t_k}} + [F\F{\widehat{\Pi}_{t_{k-1}^+}} + R + G\F{\widehat{\Lambda}_{t_{k-1}^+}}] \Delta \F{\widehat{W}_\varepsilon}(k),\\
		0 &= \lim_{t\rightarrow\infty}[C\F{\widehat{\Pi}_t} + Q + D\F{\widehat{\Lambda}_t}],
	\end{aligned}
	\end{equation}
	where $\Delta \F{\widehat{W}_\varepsilon}(k)$ is given by~\eqref{eq:dWest_definition_FI}. In addition, if there exists a matrix $K$ such that condition~\textbf{(S$_A^F$)} is satisfied, then under the control law~\eqref{eq:static_regulator_approx} with $\F{\widehat{\Gamma}_t} = \F{\widehat{\Lambda}_t} - K\F{\widehat{\Pi}_t}$, the steady-state response of the tracking error of the closed-loop system is bounded and given by
	\begin{equation}
		\label{eq:tracking_error_steady_state_FI}
		e^{ss}_t = [C\F{\widetilde{\Pi}_t^{ss}} + Q + D\F{\widehat{\Lambda}_t^{ss}}]e^{St}\omega_0,
	\end{equation}
	where $\F{\widetilde{\Pi}_t^{ss}}\in \mathbb{R}^{n\times\nu}$ is the steady-state response of
	\begin{multline}
	\label{eq:true_steady_state}
	d\F{\widetilde{\Pi}_t} = [(A+BK)\F{\widetilde{\Pi}_t}-\F{\widetilde{\Pi}_t} S+P+B\F{\widehat{\Gamma}_t}]dt + \\ [(F+GK)\F{\widetilde{\Pi}_t}+R+G\F{\widehat{\Gamma}_t}]d\mathcal{W}_t.
	\end{multline}
\end{lemma}
\begin{proof}
	By Proposition~\ref{proposition:solvability_regulator_equations} there exist bounded matrices~$\Pi_t$ and $\Lambda_t$ solving~\eqref{eq:IFI_regulator_equations}. This implies that $\F{\widehat{\Pi}_t}$ and $\F{\widehat{\Lambda}_t}$ solving~\eqref{eq:steady_state_regulation_FI} exist and are bounded. To show this, first define the auxiliary matrices $\Pi_{t_k}^D$ and $\Lambda_{t_k}^D$ satisfying
	\begin{equation}
	\label{eq:steady_state_regulation_ideal_discretised_FI}
		\begin{aligned}
		\Pi_{t_k}^D\! &=\! \Pi_{t_{k-1}}^D \!+\! [A\Pi_{t_{k-1}}^D \!- \!\Pi_{t_{k-1}}^D S \!+\! P\! +\! B\Lambda_{t_{k-1}}^D]\varepsilon +\\
		& \qquad[F\Pi_{t_{k-1}}^D+R+G\Lambda_{t_{k-1}}^D]\Delta W_\varepsilon(k),\\
		0 &= \lim_{k\rightarrow\infty}\left[C\Pi_{t_{k}}^D + Q + D\Lambda_{t_{k}}^D\right],
		\end{aligned}
	\end{equation}
	and observe that they are the solution of the forward-Euler discretisation of equations~\eqref{eq:IFI_regulator_equations} with sampling time $\varepsilon$. Therefore, let $\varepsilon \in \mathbb{R}_{>0}$ be sufficiently small as to guarantee that $\Pi_{t_k}^D$ and $\Lambda_{t_{k}}^D$ are bounded (this is possible because $\Pi_t$ and $\Lambda_t$ are bounded). Now, using the forward-Euler discretisation scheme with step $\varepsilon$, we can approximate the value of $\F{\widehat{\Pi}_t}$ in~\eqref{eq:steady_state_regulation_FI} at time $t_k$ as
	\begin{equation}
	\F{\widehat{\Pi}_{t_k}} =\F{\widehat{\Pi}_{t_{k-1}^+}}\! +\! [A\F{\widehat{\Pi}_{t_{k-1}^+}} - \F{\widehat{\Pi}_{t_{k-1}^+}} S + P + B\F{\widehat{\Lambda}_{t_{k-1}^+}}]\varepsilon,
	\end{equation}
	and, therefore, substituting this expression in the second equation in~\eqref{eq:steady_state_regulation_FI} yields
	\begin{multline}
	\label{eq:steady_state_regulation_approx_discretised_FI}
	\F{\widehat{\Pi}_{t_k^+}} = \F{\widehat{\Pi}_{t_{k-1}^+}} + [A\F{\widehat{\Pi}_{t_{k-1}^+}} - \F{\widehat{\Pi}_{t_{k-1}^+}} S + P + B\F{\widehat{\Lambda}_{t_{k-1}^+}}]\varepsilon +\\
	[F\F{\widehat{\Pi}_{t_{k-1}^+}} + R + G\F{\widehat{\Lambda}_{t_{k-1}^+}}] \Delta \F{\widehat{W}_\varepsilon}(k).
	\end{multline}
	We conclude that the discretisation of~\eqref{eq:steady_state_regulation_FI} tends to the discretisation of~\eqref{eq:IFI_regulator_equations} as $\varepsilon$ tends to zero, \textit{i.e.} they have the same forward-Euler discretisation. Recall that if the discretised system obtained using the forward-Euler scheme with sufficiently small sampling time $\varepsilon$ has bounded solutions, then the underlying continuous-time system has bounded solutions. Therefore, since $\Pi_{t_k}^D$ and $\Lambda_{t_k}^D$ are bounded, then $\F{\widehat{\Pi}_{t_k}}$ and $\F{\widehat{\Lambda}_{t_k}}$ are bounded for the same choice of $\varepsilon$ and, therefore, the original equations~\eqref{eq:steady_state_regulation_FI} have bounded solutions $\F{\widehat{\Pi}_t}$ and $\F{\widehat{\Lambda}_t}$.
	
	In turn, $\F{\widehat{\Gamma}_t} = \F{\widehat{\Lambda}_t} - K\F{\widehat{\Pi}_t}$ is bounded, therefore $\F{\widetilde{\Pi}_t}$ solving equation~\eqref{eq:true_steady_state} is bounded. As a consequence of Lemma~\ref{lemma:general_steady_state},  when the control law~\eqref{eq:static_regulator_approx} with $\F{\widehat{\Gamma}_t} = \F{\widehat{\Lambda}_t} - K\F{\widehat{\Pi}_t}$ is employed, the steady-state response of the state of the system is $x_{t}^{ss}=\F{\widetilde{\Pi}_t^{ss}}\omega$ and that of the tracking error is
	\begin{equation}
			e^{ss}_t \!\!=\!\! [(C+DK)\F{\widetilde{\Pi}_t^{ss}} + Q + D\F{\widehat{\Gamma}_t^{ss}}]\omega(t) \!=\!\! [C\F{\widetilde{\Pi}_t^{ss}} + Q + D\F{\widehat{\Lambda}_t^{ss}}]\omega(t).
	\end{equation}
	Since $\F{\widetilde{\Pi}_t^{ss}}$, $\F{\widehat{\Lambda}_t^{ss}}$ and $\omega$ are bounded, then $e^{ss}_t$ is as well.
\end{proof}

We are now ready to present the solution of the $\varepsilon$-approximate full-information output regulation problem.

\begin{theorem}
	\label{theorem:AFI_solution}
	Under Assumptions~\ref{assumption:ex_system_marg_stable}, \ref{assumption:system_stabilisable}, \ref{assumption:nonresonance} and \ref{assumption:v_not_zero_FI}, Problem~\ref{problem:epsAFI} is solvable by the control law~\eqref{eq:static_regulator_approx}.
\end{theorem}
\begin{proof}
Let $K$ be any matrix such that condition~\textbf{(S$_A^F$)} is satisfied and select \F{\widehat{\Gamma}_t} as in Lemma~\ref{lemma:steady_state_approximated}. Then the steady-state response of the tracking error of the closed-loop system is bounded and given by~\eqref{eq:tracking_error_steady_state_FI}. Recall that the forward-Euler discretisation with step $\varepsilon$ of the dynamics of $\F{\widehat{\Pi}_t}$ is given by~\eqref{eq:steady_state_regulation_approx_discretised_FI}. Using the results of Lemma \ref{lemma:brownian_motion_estimation_FI} and It\^{o}'s interpretation of the stochastic integral, we have
\begin{equation}
	\lim_{\varepsilon\rightarrow 0} \F{\widehat{\Pi}_t} = \Pi_t,\qquad  \lim_{\varepsilon\rightarrow 0} \F{\widehat{\Lambda}_t} = \Lambda_t
\end{equation}
almost surely, where the matrices $\Pi_t\in\mathbb{R}^{n\times\nu}$ and $\Lambda_t\in\mathbb{R}^{1\times\nu}$ satisfy~\eqref{eq:IFI_regulator_equations}.
As a consequence, $\lim_{\varepsilon\rightarrow 0}\F{\widetilde{\Pi}_t}=\Pi_t$ 
holds almost surely as well. Hence,
\begin{equation}
	\lim_{\varepsilon\rightarrow 0} e_t^{ss}(\omega_0,\varepsilon) = [C\Pi_t^{ss} + Q + D\Lambda_t^{ss}]e^{St}\omega_0 
\end{equation}
almost surely and, by the second condition in~\eqref{eq:IFI_regulator_equations},
\begin{equation}
	\lim_{\varepsilon\rightarrow 0} e^{ss}_t(\omega_0,\varepsilon) = 0 \quad \text{almost surely},\qquad \forall \omega_0 \in \mathbb{R}^\nu.
\end{equation}
Moreover, since Assumption \ref{assumption:ex_system_marg_stable} holds, the matrix~$e^{St}$ is bounded for all $t\in \mathbb{R}_{\ge0}$. Then
	\begin{equation}
		\lim_{\Vert \omega_0\Vert \rightarrow 0}  e^{ss}_t(\omega_0, \varepsilon) = 0\quad \text{almost surely,}  \qquad \forall \varepsilon \in \mathbb{R}_{>0}.
\end{equation}
 Therefore, condition \textbf{(R$_A^F$)} is satisfied and the $\varepsilon$-approximate full-information problem is solved.
\end{proof}

\begin{remark}
Assumption~\ref{assumption:v_not_zero_FI} amounts to requiring a persistence of excitation condition on the diffusion coefficient of the stochastic differential equation in~\eqref{eq:general_system}. This assumption is without loss of generality. In fact, if Assumption~\ref{assumption:v_not_zero_FI} did not hold, then it would be possible to choose a small enough $\bar \delta > 0$ such that $\vert v(k) \vert \le \bar \delta$ for some $k$ with nonzero probability. If this happened, the system at such time $t_k$ would be behaving as approximately deterministic as the diffusion term of the stochastic differential equation would almost be zero. Therefore, it is possible to avoid performing the stochastic compensation at $t_k$ requiring the pseudo-inversion of $v(k)$ while still obtaining satisfactory regulation performances.
\end{remark}

\begin{remark}
	If $E[\Delta \F{\widehat{W}_\varepsilon}(k)]$ were zero, then from~\eqref{eq:steady_state_regulation_FI} it would follow that $E[\F{\widehat{\Pi}_t}] = E[\F{\widetilde{\Pi}_t}] = E[\Pi_t]$ and $E[\F{\widehat{\Lambda}_t}]= E[\Lambda_t]$. Consequently, from~\eqref{eq:tracking_error_steady_state_FI} it would follow that $E[e^{ss}_t] = 0$ for all $\varepsilon \in \mathbb{R}_{>0}$. However, since~\eqref{eq:relation_DeltaW_true_est_FI} holds, then $E[\Delta \F{\widehat{W}_\varepsilon}(k)] = E[v(k-1)^+o(\varepsilon^2)]$ which is non-zero for almost all $k\in\mathbb{Z}_{> 0}$. Hence, any approximation scheme based on the forward-Euler method does not necessarily yield a steady-state tracking error with zero mean for any $\varepsilon\in \mathbb{R}_{>0}$.
\end{remark}

\section{Ideal Output-Feedback Problem}
\label{section:ideal_ef}
In this section we formulate and provide the solution of the ideal output regulation problem for system~\eqref{eq:general_system} when the state is not available for measure, but measurement outputs are available and given by\footnote{The necessity of a second measurement output is explained later in Remark~\ref{remark:two_measurement_outputs}.}
\begin{equation}
\label{eq:measurement_outputs}
y^a_t = C_a x_t, \quad y^b_t = C_b x_t,
\end{equation}
with $y^a_t\in \mathbb{R}$, $y^b_t\in \mathbb{R}$, $C_a\in \mathbb{R}^{1\times n}$, $C_b \in \mathbb{R}^{1\times n}$ and $C_a$ and $C_b$ are assumed to be linearly independent row vectors. Although not necessary (see Section~\ref{section:internal_model}), for simplicity we assume that the state~$\omega$ of the exogenous system is available for measure.

\begin{problem}
	\label{problem:IEF}
	(Ideal Output-Feedback Output Regulation Problem) Consider system~\eqref{eq:general_system}, driven by the signal generator~\eqref{eq:exogenous_system}. The \emph{ideal output-feedback output regulation problem} consists in determining a regulator such that the following conditions hold.
	\begin{description}
		\item[\textbf{(S$_I^O$)}] The closed-loop system obtained by interconnecting system~\eqref{eq:general_system} and the regulator with $\omega \equiv 0$ is asymptotically stable.
		\item[\textbf{(R$_I^O$)}] The closed-loop system obtained by interconnecting system~\eqref{eq:general_system}, the signal generator~\eqref{eq:exogenous_system} and the regulator satisfies $\lim_{t\rightarrow \infty}e_t = 0$ almost surely for any $(x_0,z_0,\omega_0)\in\mathbb{R}^{n}\times\mathbb{R}^{n_z}\times \mathbb{R}^{\nu}$.
	\end{description}
\end{problem}

Due to space limitations we report only the solution of Problem~\ref{problem:IEF} without proofs. The theoretical development is analogous to the full-information case.

The dynamic regulator\footnote{This regulator is obviously unrealistic as $\mathcal{W}_t$, which is unknown, appears explicitly. However, once again, the ideal solution is instrumental for the development of a causal solution.} of the form
	\begin{equation}
	\label{eq:ideal_regulator}
	\begin{aligned}
	dz_t &= (\mathcal{G}^{z_1} z_t + \mathcal{G}^{\omega_1}_t\omega + \mathcal{G}^1y^a_t)dt +(\mathcal{G}^{z_2}z_t + \mathcal{G}^{\omega_2}_t\omega)d\mathcal{W}_t,\\
	u_t &= Kz_t + \Gamma_t\omega,
	\end{aligned}
	\end{equation}
	where
\begin{gather}
	\mathcal{G}^{z_1} = A + BK + LC_a, \quad \mathcal{G}^1 = -L, \quad \mathcal{G}^{\omega_1}_t = P+B\Gamma_t,\\\label{eq:ideal_regulator_choice}
	\mathcal{G}^{z_2} = F + GK, \quad \mathcal{G}^{\omega_2}_t = R+G\Gamma_t,
\end{gather}
with $K$, if it exists, such that system~\eqref{eq:general_system} with $u_t = Kx_t$ and $\omega \equiv 0$ is asymptotically stable, $L \in \mathbb{R}^{n\times 1}$, if it exists, such that system
$
	d\eta_t = (A+LC_a)\eta_t dt + F\eta_t d\mathcal{W}_t
$ 
is asymptotically stable and $\Gamma_t=\Lambda_t - K\Pi_t$, with $\Pi_t$ and $\Lambda_t$ bounded solutions of equations~\eqref{eq:IFI_regulator_equations}, solves Problem~\ref{problem:IEF}.
\begin{remark}
    This ideal solution is based on the separation principle. It is well known that the separation principle only holds under the assumption that the Brownian motion is available for feedback, which is impossible in practice. Thus, the approximate output-feedback solution in the next section cannot be based on the separation principle.
\end{remark}

\section{Approximate Output-Feedback Problem}
\label{section:approximate_ef}

For analogous reasons to those reported at the beginning of Section~\ref{section:approximate_fi}, the solution of the ideal output-feedback problem cannot be implemented in practice. In fact, besides being essential in the integration of the regulator equations~\eqref{eq:IFI_regulator_equations}, the knowledge of the signal $\mathcal{W}_t$ is also needed to implement the dynamic regulator~\eqref{eq:ideal_regulator}. In this section we provide a hybrid control architecture that solves a weaker version of the output-feedback problem. In particular, first we define an approximate problem which we aim at solving using a hybrid regulator that employs estimates of the variations of the Brownian motion; then we describe the steady state of the resulting hybrid closed-loop system, we characterise how the Brownian motion is reconstructed, we describe the resulting hybrid estimator and we provide the solution to the approximate problem. 
\begin{problem}
	\label{problem:AEF}
	($\varepsilon$-Approximate Output-Feedback Output Regulation Problem). Consider system~\eqref{eq:general_system} driven by the signal generator~\eqref{eq:exogenous_system}. The $\varepsilon$-\emph{approximate output-feedback output regulation problem} consists in determining a  regulator such that the following conditions hold.
	\begin{description}
		\item[\textbf{(S$_A^O$)}] The closed-loop system obtained by interconnecting system~\eqref{eq:general_system} and the regulator with $\omega \equiv 0$ is asymptotically stable.
		\item[\textbf{(R$_A^O$)}] The closed-loop system obtained by interconnecting system~\eqref{eq:general_system}, the signal generator~\eqref{eq:exogenous_system} and the regulator yields a steady-state response of the tracking error $e^{ss}_t(\omega_0, \varepsilon)$, with $\varepsilon\in \mathbb{R}_{>0}$, which is bounded and such that
		\begin{equation}
		\begin{aligned}
		\lim_{\Vert \omega_0 \Vert \rightarrow 0} e_t^{ss}(\omega_0, \varepsilon) &= 0,\quad  \forall  \varepsilon \in \mathbb{R}_{>0},\\
		\lim_{\varepsilon\rightarrow 0}  e_t^{ss}(\omega_0, \varepsilon) &= 0,\quad  \forall  \omega_0 \in \mathbb{R}^\nu,
		\end{aligned}
		\end{equation}
		almost surely, for any $(x_0,z_0)\in\mathbb{R}^{n\times n_z}$.
	\end{description}
\end{problem}

With the solution of the ideal problem at hand (equation~\eqref{eq:ideal_regulator}), we now consider the following jump system
\begin{equation}
\label{eq:hybrid_regulator}
\begin{aligned}
dz_t &= \left[\mathcal{G}^{z_1} z_t + \mathcal{G}^1 y^a_t + \widehat{\mathcal{G}}^{\omega_1}_t \omega\right]dt,\\
z_{t_k^+} &=  z_{t_k} + \left[\mathcal{G}^{z_2} z_{t_{k-1}^+}  + \widehat{\mathcal{G}}^{\omega_2}_{t_{k-1}^+}\omega(t_{k-1}^+)\right]\Delta \widehat{W}_\varepsilon(k), \\
u_t &= K z_t + \E{\widehat{\Gamma}_t} \omega,
\end{aligned}
\end{equation}
where $\mathcal{G}^{z_1}$, $\mathcal{G}^1$, $\mathcal{G}^{z_2}$ and $K$ have the same meaning as in the regulator~\eqref{eq:ideal_regulator}-\eqref{eq:ideal_regulator_choice}, whereas $\widehat{\mathcal{G}}^{\omega_1}_t\in\mathbb{R}^{n_z\times \nu}$, $\widehat{\mathcal{G}}^{\omega_2}_{t_{k}}\in\mathbb{R}^{n_z\times \nu}$ and $\E{\widehat{\Gamma}_t} \in \mathbb{R}^{1\times \nu}$ are bounded. Again, the parameter $\varepsilon = t_k - t_{k-1}$ for all $k\in \mathbb{Z}_{>0}$ is the sampling period, whereas $\Delta \E{\widehat{W}_\varepsilon}(k)$ is an approximation of the variation of the Brownian motion $\Delta W_\varepsilon(k) = \mathcal{W}_{t_k} - \mathcal{W}_{t_{k-1}}$. The construction of $\Delta \widehat{W}_\varepsilon(k)$ is postponed to Section~\ref{sec:estimation_brownian_motion}. Observe that the definition of the approximate output-feedback problem is analogous to the full-information counterpart. Therefore, the regulation accuracy improves as the norm of the exogenous input and/or the sampling period tend to zero.

\subsection{Steady State of the Closed-Loop System}
In this section we characterise the steady-state response of the closed-loop system obtained interconnecting~\eqref{eq:general_system}, \eqref{eq:exogenous_system} and \eqref{eq:hybrid_regulator}. In analogy with the selections~\eqref{eq:ideal_regulator_choice}, let the regulator matrices be
	\begin{gather}
	\mathcal{G}^{z_1} = A + BK + LC_a, \quad \mathcal{G}^1 = -L, \quad\widehat{\mathcal{G}}^{\omega_1}_t = P+B\E{\widehat{\Gamma}_t},\\ \label{eq:AEF_selections}
	\mathcal{G}^{z_2} = F+ GK, \qquad \widehat{\mathcal{G}}^{\omega_2}_{t_k} = R+G\E{\widehat{\Gamma}_{t_{k}}}.
	\end{gather}
 Again, with these selections the variable $z_t$ is an estimation of the state $x_t$. To write the closed-loop dynamics avoiding the explicit use of delay, we introduce the auxiliary variable $z^\ell_t = z_{t_k^+}$ for all $t\in [t_k, t_{k+1})$, which holds the value of $z_t$ between sampling times. Moreover, we also notice that since the exogenous system is deterministic and continuous, then $\omega(t_k^+) = \omega(t_k)$ and $\omega(t_{k-1}^+) = e^{-S\varepsilon}\omega(t_k)$  for all $k\in\mathbb{Z}_{>0}$.

 Let $\tilde{x}_t = [x_t^\top\;z_t^\top\; z^{\ell \top}_t]^\top$. Then the closed-loop system obtained interconnecting systems~\eqref{eq:general_system}, \eqref{eq:exogenous_system} and \eqref{eq:hybrid_regulator} with the selections~\eqref{eq:AEF_selections} has the dynamics
\begin{equation}
\label{eq:jump_closed_loop_zed}
\begin{aligned}
d\tilde{x}_t &= (\tilde{A}^c\tilde{x}_t + \tilde{P}^c_t\omega)dt + (\tilde{F}^c\tilde{x}_t + \tilde{R}^c_t\omega)d\mathcal{W}_t,\\
\tilde{x}_{t_k^+} &= \tilde{A}^d\tilde{x}_{t_k} + (\tilde{F}^d\tilde{x}_{t_k} + \tilde{R}^d_{t_k}\omega(t_k))\Delta \E{\widehat{W}_\varepsilon}(k),
\end{aligned}
\end{equation}
with
\begin{equation}
\tilde{A}^c = \begin{bmatrix}
A & BK  & 0\\ -LC_a & A+BK+LC_a & 0\\ 0 & 0 & 0
\end{bmatrix},\;\tilde{P}^c_t = \begin{bmatrix}
P + B\E{\widehat{\Gamma}_t} \\ P + B\E{\widehat{\Gamma}_t} \\ 0 
\end{bmatrix},
\end{equation}
\begin{equation}
\tilde{F}^c\! =\! \begin{bmatrix}
F & GK & 0\\ 0 & 0 & 0 \\ 0 & 0 & 0
\end{bmatrix}\!, \,
\tilde{R}^c_t\! =\! \begin{bmatrix}
R+G\E{\widehat{\Gamma}_t} \\ 0 \\ 0
\end{bmatrix}\!, \,\tilde{A}^d\! =\! \begin{bmatrix}
I & 0 & 0 \\ 0 & I & 0 \\ 0 & I & 0
\end{bmatrix}\!,
\end{equation}
\begin{equation}
\tilde{F}^d = \begin{bmatrix}
0 & 0 & 0\\ 0 & 0 & F+GK \\ 0 & 0 & F+GK
\end{bmatrix},\;\tilde{R}^d_{t_k} = \begin{bmatrix}
 0 \\ R+G\E{\widehat{\Gamma}_{t_{k-1}^+}}\\  R+G\E{\widehat{\Gamma}_{t_{k-1}^+}}
\end{bmatrix}e^{-S\varepsilon}.
\end{equation}

In the following results we derive the steady-state response of the state of the closed-loop system~\eqref{eq:jump_closed_loop_zed}. These will be used later to characterise the properties of the steady-state tracking error.
\begin{lemma}
	\label{lemma:steady_state_hybrid}
	Consider system~\eqref{eq:jump_closed_loop_zed}. Assume that there exist matrices $K$ and $L$ such that condition \textbf{(S$_A^O$)} is satisfied. Then the steady-state response of the state $\tilde x_t$ is $\tilde{x}^{ss}_t = \widetilde{\mathcal{X}}_t^{ss}\omega(t)$, where $\widetilde{\mathcal{X}}^{ss}_t\in \mathbb{R}^{3n\times \nu}$ is the steady-state response of $\widetilde{\mathcal{X}}_t\in \mathbb{R}^{3n\times \nu}$, solution of
	\begin{equation}
	\label{eq:steady_state_hybrid}
	\begin{aligned}
	d\widetilde{\mathcal{X}}_t &= (\tilde{A}^c\widetilde{\mathcal{X}}_t - \widetilde{\mathcal{X}}_tS + \tilde{P}^c_t)dt + (\tilde{F}^c\widetilde{\mathcal{X}}_t + \tilde{R}^c_t)d\mathcal{W}_t,\\
	\widetilde{\mathcal{X}}_{t_k^+} &= \tilde{A}^d\widetilde{\mathcal{X}}_{t_k} + (\tilde{F}^d\widetilde{\mathcal{X}}_{t_k} + \tilde{R}^d_{t_k})\Delta \E{\widehat{W}_\varepsilon}(k).
	\end{aligned}
	\end{equation}
\end{lemma}
\begin{proof}
	Define the variable $\chi_t = \tilde{x}_t - \widetilde{\mathcal{X}}_t\omega(t)$. Then observe that the dynamics of $\chi_t$ is given by
	\begin{equation}
	\label{eq:steady_state_convergence}
		\begin{aligned}
		d\chi_t &= \tilde{A}^c\chi_tdt + \tilde{F}^c\chi_td\mathcal{W}_t,\\
		\chi_{t_k^+} &= \tilde{A}^d\chi_{t_k} + \tilde{F}^d\chi_{t_k} \Delta \E{\widehat{W}_\varepsilon}(k).
		\end{aligned}
	\end{equation}
	Since \textbf{(S$_A^O$)} is satisfied by hypothesis, system~\eqref{eq:steady_state_convergence} is asymptotically stable. Then $\lim_{t\rightarrow \infty} \chi_t = 0$ almost surely, hence the claim follows.
\end{proof}
The following corollary characterises the steady-state responses of $x_t$ and $z_t$. This result will be used in the following to derive the solution of Problem~\ref{problem:AEF}.
\begin{corollary}
	\label{corollary:steady_state_partitioned}
	Consider system~\eqref{eq:jump_closed_loop_zed}. Assume that there exist matrices $K$ and $L$ such that condition \textbf{(S$_A^O$)} is satisfied. Then the steady-state responses of $x_t$ and $z_t$ are $x^{ss}_t = \widetilde{\Pi}^{x,ss}_t\omega(t)$ and $z^{ss}_t = \widetilde{\Pi}^{z,ss}_t\omega(t)$, respectively, where $\widetilde{\Pi}^{x,ss}_t \in \mathbb{R}^{n\times\nu}$ and $\widetilde{\Pi}^{z,ss}_t \in \mathbb{R}^{n\times\nu}$ are the steady-state responses of
	\begin{equation}
	\label{eq:steady_state_partitioned}
		\begin{aligned}
		d\widetilde{\Pi}^x_t &= \left[A\widetilde{\Pi}^x_t +BK\widetilde{\Pi}^z_t - \widetilde{\Pi}^x_t S + P + B\E{\widehat{\Gamma}_t}\right]dt+\\ & \quad\quad \quad \qquad \qquad\left[F\widetilde{\Pi}^x_t + GK\widetilde{\Pi}^z_t + R + G\E{\widehat{\Gamma}_t}\right] d\mathcal{W}_t,\\
		d\widetilde{\Pi}^z_t &\!=\!\! \left[-LC_a\widetilde{\Pi}^x_t\! +\! (A\!+\!BK\!+\!LC_a)\widetilde{\Pi}^z_t\! -\! \widetilde{\Pi}^z_{t}S\! +\! P\! +\! B\E{\widehat{\Gamma}_t}\right]\!\!dt,\\
		\widetilde{\Pi}^x_{t_k^+} &= \widetilde{\Pi}^x_{t_k},\\
		\widetilde{\Pi}^z_{t_k^+}\! &= \widetilde{\Pi}^z_{t_k}\! + \!\left[(F\!+\!GK)\widetilde{\Pi}^\ell_{t_k} \!+\! (R\!+\!G\E{\widehat{\Gamma}_{t_{k-1}^+}})e^{-S\varepsilon}\right]\Delta \E{\widehat{W}_\varepsilon}(k),
		\end{aligned}
	\end{equation}
	where the matrix $\widetilde{\Pi}^\ell_{t}$ solves the auxiliary equations
	\begin{equation}
		\begin{aligned}
		d\widetilde{\Pi}^\ell_{t} &= -\widetilde{\Pi}^\ell_{t}Sdt,\\
		\widetilde{\Pi}^\ell_{t_k^+}\! &= \widetilde{\Pi}^z_{t_k}\! +\! \left[(F\!+\!GK)\widetilde{\Pi}^\ell_{t_k}\! + \!(R\!+\!G\E{\widehat{\Gamma}_{t_{k-1}^+}})e^{-S\varepsilon}\right]\Delta \E{\widehat{W}_\varepsilon}(k).
		\end{aligned}
	\end{equation}
\end{corollary}
\begin{proof}
	The claim follows by partitioning $\widetilde{\mathcal{X}}_t$ in~\eqref{eq:steady_state_hybrid} as $\widetilde{\mathcal{X}}_t = [\widetilde{\Pi}^{x \top}_t\; \widetilde{\Pi}^{z \top}_t \; \widetilde{\Pi}^{\ell \top}_t]^\top$ and using the result of Lemma~\ref{lemma:steady_state_hybrid}.
\end{proof}

\subsection{Reconstruction of the Brownian Motion}
\label{sec:estimation_brownian_motion}
We now show how the sequence of scalars $\{\Delta \E{\widehat{W}_\varepsilon}(k)\}_{k>0}$ is constructed. In particular, the scalars approximate \emph{a posteriori} the variations of the Brownian motion with a degree of accuracy that depends on $\varepsilon$. To this end, we first define the following condition.

\textbf{(EC)} The closed-loop system obtained by interconnecting system~\eqref{eq:general_system}, the signal generator~\eqref{eq:exogenous_system} and the regulator~\eqref{eq:hybrid_regulator} satisfies $\lim_{t\rightarrow \infty} (z_t - x_t) = o(\varepsilon^2)$ almost surely.

This condition is roughly equivalent to assuming that the regulator~\eqref{eq:hybrid_regulator} contains an observer for the state $x_t$. For the time being we assume that \textbf{(EC)} is satisfied and we discuss how to enforce this condition in the next section. For ease of notation, define the quantities
\begin{equation}
	\begin{aligned}
	v^x(k) &= C_b(Fx_{t_k} + Gu_{t_{k}} + R\omega(t_k)), \quad \forall k \in \mathbb{Z}_{\ge 0},\\
	v^z(k) &= C_b(Fz_{t_k} + Gu_{t_k} + R\omega(t_k)), \quad \forall k \in \mathbb{Z}_{\ge 0},
	\end{aligned}
\end{equation}
and $\Delta y^b(k) = y^b_{t_k} - y^b_{t_{k-1}}$. Note that $v^x(k)$ is the diffusion term of the output dynamics $dy^b_t = C_bdx_t$ evaluated at time $t_k$ and $v^z(k)$ is its approximation when the state $x_t$ is replaced by its estimate $z_t$.
As discussed in Section~\ref{subsec:reconstruction_brownian_motion_FI}, it is again reasonable to assume the following.
\begin{assumption}
\label{assumption:v_not_zero_EF}
There exist $\delta_x \in \mathbb{R}_{>0}$ and $\delta_z \in \mathbb{R}_{>0}$ such that $\vert v^x(k) \vert > \delta_x$ and $\vert v^z(k) \vert > \delta_z$ almost surely for all $k\in \mathbb{Z}_{\ge 0}$.
\end{assumption}
The converse, in fact, would imply that the output $y^b_t$ would not display stochastic dynamics for some $t_k$ with nonzero probability, which would make it impossible to reconstruct the Brownian motion from its measures.

We are now ready to give a constructive result on the sequence $\{\Delta \widehat{W}_\varepsilon(k)\}_{k>0}$.

\begin{lemma}
	\label{lemma:brownian_motion_estimation_EF}
	Consider system~\eqref{eq:general_system} and the regulator~\eqref{eq:hybrid_regulator} and let Assumption~\ref{assumption:v_not_zero_EF} hold. Assume that condition \textbf{(EC)} holds. Let the sequence $\{\Delta \E{\widehat{W}_\varepsilon}(k)\}_{k>0}$ be defined as
	\begin{multline}
	\label{eq:dWest_definition_EF}
		\Delta \E{\widehat{W}_\varepsilon}(k) = v^z(k-1)^{-1} [\Delta y^b(k) -\\ C_b(Az_{t_{k-1}} + Bu_{t_{k-1}} + P\omega(t_{k-1}))\varepsilon].
	\end{multline}
	Then $\lim_{k\rightarrow \infty}\Delta \E{\widehat{W}_\varepsilon}(k) \xrightarrow{\varepsilon} d\mathcal{W}_t$ almost surely.
\end{lemma}
\begin{proof}
	Let $k\in\mathbb{Z}_{>0}$. By \cite[Theorem 7.1]{Gard1988}
	\begin{multline}
	\label{eq:first_order_approx_y2}
		\Delta y^b(k) = C_b(Ax_{t_{k-1}} + Bu_{t_{k-1}} + P\omega(t_{k-1}))\varepsilon +\\
		v^x(k-1) \Delta W_\varepsilon(k) + o(\varepsilon^2)
	\end{multline}
	holds, where $o(\varepsilon^2)$ is the \emph{one-step truncation error} of the forward-Euler discretisation scheme. Since $v^x(k)\ne 0$ almost surely, the expression
	\begin{multline}
	\Delta W_\varepsilon(k) = v^x(k-1)^{-1}[\Delta y^b(k) -\\ C_b(Ax_{t_{k-1}} + Bu_{t_{k-1}} + P\omega(t_{k-1}))\varepsilon + o(\varepsilon^2)]
	\end{multline}
	holds almost surely. Let
	\begin{multline}
	    \Delta \widehat{W}^{*}_\varepsilon(k) = (v^x(k-1)+ o(\varepsilon^2))^{-1} [\Delta y^b(k) -\\ C_b(Ax_{t_{k-1}} + Bu_{t_{k-1}} + P\omega(t_{k-1}))\varepsilon + o(\varepsilon^2)].
	\end{multline}
	Since $\lim_{t\rightarrow \infty}(z_t - x_t) = o(\varepsilon^2)$ almost surely, then
	\begin{equation}
	\label{eq:relation_DeltaW_true_est_EF}
	\lim_{k\rightarrow \infty}\left(\Delta \E{\widehat{W}_\varepsilon}(k) - \Delta \widehat{W}^{*}_\varepsilon(k)\right) = 0,
	\end{equation}
	holds almost surely. Simple computations yield
        $
	    \Delta \widehat{W}^{*}_\varepsilon(k) = \Delta W_\varepsilon(k) + \rho(k,\varepsilon)o(\varepsilon^2),
	    $
	with
	\begin{multline}
	    \rho(k,\varepsilon) = -[v^x(k-1)(v^x(k-1) + o(\varepsilon^2))]^{-1} 
	    [v^x(k-1) +\\\Delta y^b(k) - C_b(Ax_{t_{k-1}} + Bu_{t_{k-1}} + P\omega(t_{k-1}))\varepsilon + o(\varepsilon^2)],
	\end{multline}
	and note that $\lim_{\varepsilon\rightarrow 0}\rho(k,\varepsilon)o(\varepsilon^2) = 0$. The rest of the proof is analogous to that of Lemma~\ref{lemma:brownian_motion_estimation_FI} with $\rho(k,\varepsilon)$ playing the role of $v(k-1)^+$.
\end{proof}

\subsection{Hybrid State Estimator}
The problem we now need to address is the satisfaction of condition \textbf{(EC)} (without assuming it). If $K$ and $L$ are such that condition \textbf{(S$_A^O$)} holds, then both $x_t$ and $z_t$ converge to zero almost surely when $\omega \equiv 0$. The case $\omega \ne 0$, however, requires more care. Indeed, a non-zero exogenous input forces the state of the system to lie on a non-zero manifold at steady state. We show that, using the regulator~\eqref{eq:hybrid_regulator}, it is possible to obtain an estimate $z_t$ that converges to the actual state $x_t$ as the sampling period $\varepsilon$ tends to zero. Note that the satisfaction of this property is necessary for Lemma~\ref{lemma:brownian_motion_estimation_EF} to hold. We now characterise the discrete-time dynamics of the estimation error and discuss the choice of $K$ and $L$ such that \textbf{(EC)} holds.

\begin{lemma}
	\label{lemma:est_error_hybrid}
	Let Assumption \ref{assumption:v_not_zero_EF} hold. The forward-Euler discretisation of the estimation error dynamics $\eta_t$ is
	\begin{multline}
	\label{eq:est_error_dyn_nonlin}
	\eta_{t_k} = \eta_{t_{k-1}} + [(I-\mygamma_{t_{k-1}})A+LC_a]\eta_{t_{k-1}}\varepsilon + \\(I-\mygamma_{t_{k-1}})F\eta_{t_{k-1}}\Delta W_\varepsilon(k) +o(\varepsilon^2),
	\end{multline}
	where
	\begin{equation}
	\label{eq:mygamma_definition}
	\mygamma_{t_{k}} \!=\! ((F+GK)z_{t_{k}}\! +\! (R\!+\!G\Gamma_{t_{k}})\omega(k))v^{z}(k)^{-1}C_b.
	\end{equation}
\end{lemma}
\begin{proof}
Express the dynamics of the closed-loop system~\eqref{eq:jump_closed_loop_zed} in terms of the state vector $\hat{x}_t = [x_t^\top\;\eta_t^\top\; z^{\ell \top}_t]^\top$, \emph{i.e.} replacing $z_t^\top$ with the error variable $\eta_t^\top$. This yields
\begin{equation}
\label{eq:jump_closed_loop}
\begin{aligned}
d\hat{x}_t &= (\hat{A}^c\hat{x}_t + \hat{P}^c_t\omega)dt + (\hat{F}^c\hat{x}_t + \hat{R}^c_t\omega)d\mathcal{W}_t,\\
\hat{x}_{t_k^+} &= \hat{A}^d\hat{x}_{t_k} + (\hat{F}^d\hat{x}_{t_k} + \hat{R}^d_{t_k}\omega(t_k))\Delta \E{\widehat{W}_\varepsilon}(k),
\end{aligned}
\end{equation}
with
\begin{gather}
\hat{A}^c\!\! =\!\! \begin{bmatrix}
\!A+BK\!\! & -BK  & 0\\ 0 & \!\!A+LC_a\!\! & 0\\ 0 & 0 & 0
\end{bmatrix}\!\!,\, \hat{F}^c\! =\! \begin{bmatrix}
F+GK & -GK & 0\\ F+GK & -GK & 0 \\ 0 & 0 & 0
\end{bmatrix}\!, \\ \hat{P}^c_t = \begin{bmatrix}
P + B\E{\widehat{\Gamma}_t} \\ 0  \\ 0 
\end{bmatrix},\,\hat{R}^c_t\! =\! \begin{bmatrix}
R+G\E{\widehat{\Gamma}_t} \\ R+G\E{\widehat{\Gamma}_t} \\ 0 
\end{bmatrix}\!, \,\tilde{A}^d\! =\! \begin{bmatrix}
I & 0 & 0 \\ 0 & I & 0 \\ I & -I & 0
\end{bmatrix}\!,\\
\hat{F}^d\!\! =\!\! \begin{bmatrix}
0 & 0 & 0\\ 0 & 0 & \!\!-(F+GK) \!\\ 0 & 0 & F+GK
\end{bmatrix}\!\!,\;\tilde{R}^d_{t_k}\! =\!\! \begin{bmatrix}
0 \\  -(R+G\E{\widehat{\Gamma}_{t_{k-1}^+}})\\  R+G\E{\widehat{\Gamma}_{t_{k-1}^+}}
\end{bmatrix}e^{-S\varepsilon}.
\end{gather}
Observe that the forward-Euler discretisation of the continuous-time dynamics, for sufficiently small $\varepsilon$, is given by
\begin{multline}
\hat{x}_{t_k} = \hat{x}_{t_{k-1}^+} + (\hat{A}^c\tilde{x}_{t_{k-1}^+} + \hat{P}^c_{t_{k-1}^+}\omega(t_{k-1}^+))\varepsilon + \\(\hat{F}^c\hat{x}_{t_{k-1}^+} + \hat{R}^c_{t_{k-1}^+}\omega(t_{k-1}^+))\Delta W_\varepsilon(k) + o(\varepsilon^2);
\end{multline}
substituting this expression of $\hat{x}_{t_k}$ in the second equation in~\eqref{eq:jump_closed_loop} we obtain the following discrete-time dynamics of the state $\xi_{t_k} = [x_{t_k}^\top \; \eta_{t_k}^\top]^\top$ (we drop the jump notation for clarity and we replace $t_k$ with just $k$ with a slight abuse of notation)
\begin{multline}
\xi_k = \xi_{k-1} + \begin{bmatrix}
A + BK & -BK \\ 0 & A+LC_a
\end{bmatrix} \xi_{k-1}\varepsilon +\\
\begin{bmatrix}
P+B\E{\widehat{\Gamma}_{k-1}}\\ 0
\end{bmatrix}\omega(k-1)\varepsilon + \begin{bmatrix}
F+GK & -GK \\ F+GK & -GK
\end{bmatrix}\xi_{k-1}\Delta W_\varepsilon(k) +\\ \begin{bmatrix}
0 & 0 \\ -(F+GK) & F+GK
\end{bmatrix}\xi_{k-1}\Delta \E{\widehat{W}_\varepsilon}(k) + \\\begin{bmatrix}
R + G\E{\widehat{\Gamma}_{k-1}} \\ R + G\E{\widehat{\Gamma}_{k-1}}
\end{bmatrix}\omega(k-1)\Delta W_\varepsilon(k)+\\ \begin{bmatrix}
0 \\ R + G\E{\widehat{\Gamma}_{k-1}}
\end{bmatrix}\omega(k-1)\Delta\E{\widehat{W}_\varepsilon}(k)+ o(\varepsilon^2).
\label{eq:intermediate_closed_loop}
\end{multline}
Note that the auxiliary variable $z^\ell_t$ is not needed anymore in the discrete-time domain, as it would be redundant. Now, define the matrix $\mygamma_{t_{k-1}}$ as in~\eqref{eq:mygamma_definition}. Observe that substituting~\eqref{eq:first_order_approx_y2} in~\eqref{eq:dWest_definition_EF} we obtain
\begin{multline}
\label{eq:deltaW_est_alternative}
\Delta \E{\widehat{W}_\varepsilon}(k) = v^z(k-1)^{-1}C_b [A(x_{t_{k-1}}-z_{t_{k-1}})\varepsilon  +\\ (Fx_{t_{k-1}} + Gu_{t_{k-1}} + R\omega(k-1))\Delta W_\varepsilon(k) + o(\varepsilon^2)].
\end{multline}
Focusing on the estimation error dynamics, substituting~\eqref{eq:deltaW_est_alternative} in~\eqref{eq:intermediate_closed_loop} and rearranging we obtain

\begin{multline}
\eta_k = \eta_{k-1} + [(I-\mygamma_{k-1})A+LC_a]\eta_{k-1}\varepsilon + (I-\mygamma_{k-1})\times\\(Fx_{k-1}+GKz_{k-1} + (R+G\E{\widehat{\Gamma}_{k-1}})\omega(k\!-\!1))\Delta W_\varepsilon(k) + o(\varepsilon^2).
\end{multline}
Adding and subtracting $Fz_{k-1}$ in the last term and observing that
$(I-\mygamma_{k-1})(Fz_{k-1}+GKz_{k-1} + (R+G\E{\widehat{\Gamma}_{k-1}})\omega(k-1)) = 0$
for all $k\in \mathbb{Z}_{>0}$, we conclude
\begin{multline}
\eta_k = \eta_{k-1} + [(I-\mygamma_{k-1})A+LC_a]\eta_{k-1}\varepsilon + \\(I-\mygamma_{k-1})F\eta_{k-1}\Delta W_\varepsilon(k).
\end{multline}
\end{proof}

\begin{remark}
\label{remark:two_measurement_outputs}
		Since $C_b(I - \mygamma_{t_k})A = 0$ for all $k\in \mathbb{Z}_{\ge 0}$, there is a loss of observability through the output $y_b$ used to reconstruct the Brownian motion. This justifies the need for two different measurement outputs and, additionally, the linear independence of the row vectors $C_a$ and $C_b$. Therefore, $y^a_t$ is used to estimate the state of the system whereas $y^b_t$ is used to approximate the variations of the Brownian motion $\Delta \E{\widehat{W}_\varepsilon} (k)$, according to expression~\eqref{eq:dWest_definition_EF}.
\end{remark}

The forward-Euler discretisation of the dynamics of the estimation error~\eqref{eq:est_error_dyn_nonlin} is nonlinear and time-varying. Moreover, the choice of the gain $L$ that stabilises the estimation error dynamics cannot be independent of the choice of the gain $K$. This is consistent with the fact that the design of the observer and of the controller cannot be separated in practice \cite[Section 1.8]{Damm2004}. Consequently, the requirements on the system stabilisability and detectability are expressed by the following assumption.

\begin{assumption}
	\label{assumption:KL_existence}
	There exist matrices $K(t)$ and $L(t)$ such that \textbf{(S$_A^O$)} and \textbf{(EC)} are satisfied.
\end{assumption}

\begin{remark}
\label{remark:choice_of_gains}
The design of $K$ and $L$ yielding \textbf{(S$_A^O$)} and \textbf{(EC)} is in general a non-trivial task. However, observing the discretised dynamics of the closed-loop system, and specifically of the estimation error~\eqref{eq:est_error_dyn_nonlin}, suggests designing the stabilising gain $K$ first and, subsequently, finding a piecewise constant stabilising $L$ for the estimation error subsystem, discretised by~\eqref{eq:est_error_dyn_nonlin}, where $K$ now appears as a parameter. This simplifies the design of $K$ and $L$ but it may result in a restrictive selection.
\end{remark}

\subsection{Solution to the Approximate Problem}
In this section we show how to solve the $\varepsilon$-approximate output feedback output regulation problem. To this end, we first give a preliminary result.
\begin{lemma}
	\label{lemma:steady_state_regulation}
	Consider the closed-loop system obtained interconnecting \eqref{eq:general_system}, \eqref{eq:exogenous_system} and \eqref{eq:hybrid_regulator} with the selections~\eqref{eq:AEF_selections} and let Assumptions~\ref{assumption:ex_system_marg_stable}, \ref{assumption:nonresonance} and \ref{assumption:v_not_zero_EF} hold. Suppose that there exist matrices $(K,L)$ such that condition \textbf{(EC)} is satisfied. Then there exist bounded matrices $\E{\widehat{\Pi}_t}\in\mathbb{R}^{n\times\nu}$ and $\E{\widehat{\Lambda}_t}\in \mathbb{R}^{1\times\nu}$ solving~\eqref{eq:steady_state_regulation_FI}, where $\Delta \E{\widehat{W}_\varepsilon}(k)$ is given by~\eqref{eq:dWest_definition_EF}. Moreover, if $(K,L)$ are such that \textbf{(S$_A^O$)} holds, then under the control law~\eqref{eq:hybrid_regulator} with $\E{\widehat{\Gamma}_t = \E{\widehat{\Lambda}_t}-K\E{\widehat{\Pi}_t}}$, the steady-state response of the tracking error of the closed-loop system is bounded and given by
	\begin{equation}
	\label{eq:tracking_error_steady_state_EF}
	e^{ss}_t = [C\widetilde{\Pi}^{x,ss}_t + DK\widetilde{\Pi}^{z,ss}_t + Q + D\E{\widehat{\Gamma}_t^{ss}}]e^{St}\omega_0,
	\end{equation}
	with $\widetilde{\Pi}^{x,ss}_t$ and $\widetilde{\Pi}^{z,ss}_t$ as defined in Corollary~\ref{corollary:steady_state_partitioned}.
\end{lemma}
\begin{proof}
By Proposition~\ref{proposition:solvability_regulator_equations} there exist bounded matrices $\Pi_t$ and $\Lambda_t$ solving~\eqref{eq:IFI_regulator_equations}. Following the procedure adopted in the proof of Lemma~\ref{lemma:steady_state_approximated}, we obtain~\eqref{eq:steady_state_regulation_approx_discretised_FI}. Since \textbf{(EC)} holds, by Lemma~\ref{lemma:brownian_motion_estimation_EF}, we conclude that the discretisation of \eqref{eq:steady_state_regulation_FI} (with $\Delta \E{\widehat{W}_\varepsilon}(k)$ given by~\eqref{eq:dWest_definition_EF}) tends to the discretisation of~\eqref{eq:IFI_regulator_equations} as $\varepsilon$ tends to zero and $k$ tends to infinity, \textit{i.e.} they have the same forward-Euler discretisation at steady state. Repeating the discussion reported in the proof of Lemma~\ref{lemma:steady_state_approximated}, we conclude that $\E{\widehat{\Pi}_t}$ and $\E{\widehat{\Lambda}_t}$ are bounded.
	The boundedness of $\E{\widehat{\Pi}_t}$ and $\E{\widehat{\Lambda}_t}$, and hence of $\E{\widehat{\Gamma}_t = \E{\widehat{\Lambda}_t}-K\E{\widehat{\Pi}_t}}$, in turn implies that the matrices $\widetilde{\Pi}^x_t$ and $\widetilde{\Pi}^z_t$ given in Corollary~\ref{corollary:steady_state_partitioned} are bounded. Moreover, by Corollary~\ref{corollary:steady_state_partitioned}, when the regulator~\eqref{eq:hybrid_regulator} is employed, the steady-state response of the state of the controlled system is $x_{t}^{ss}=\widetilde{\Pi}^{x,ss}_t\omega$ and the steady-state response of the estimate of the state is $z_{t}^{ss}=\widetilde{\Pi}^{z,ss}_t\omega$. Substituting these quantities in the expression of the tracking error $e_t$ yields its steady-state response
	$e^{ss}_t = [C\widetilde{\Pi}^{x,ss}_t + DK\widetilde{\Pi}^{z,ss}_t + Q + D\E{\widehat{\Gamma}_t^{ss}}]\omega(t)$.
	 Since $\widetilde{\Pi}^{x,ss}_t$, $\widetilde{\Pi}^{z,ss}_t$, $\E{\widehat{\Gamma}_t^{ss}}$ and $\omega$ are bounded, then $e^{ss}_t$ is as well.
\end{proof}

We are now ready to present the solution of the $\varepsilon$-approximate output-feedback output regulation problem.
\begin{theorem}
	\label{theorem:AEF_solution}
	Under Assumptions~\ref{assumption:ex_system_marg_stable}, \ref{assumption:nonresonance}, \ref{assumption:v_not_zero_EF} and \ref{assumption:KL_existence}, Problem~\ref{problem:AEF} is solvable by the regulator~\eqref{eq:hybrid_regulator} with the selections~\eqref{eq:AEF_selections}.
\end{theorem}
\begin{proof}
	Let $(K,L)$ be any pair such that conditions~\textbf{(S$_A^O$)} and \textbf{(EC)} are satisfied. By Lemma~\ref{lemma:steady_state_regulation} the steady-state response of the tracking error is bounded and given by~\eqref{eq:tracking_error_steady_state_EF}. Recall that the forward-Euler discretisation with step $\varepsilon$ of the dynamics of $\E{\widehat{\Pi}_t}$ is~\eqref{eq:steady_state_regulation_approx_discretised_FI}.
	Using the results of Lemma \ref{lemma:brownian_motion_estimation_EF} and It\^{o}'s interpretation of the stochastic integral, we have
	\begin{equation}
	\lim_{\varepsilon\rightarrow 0}\lim_{t\rightarrow \infty} \left(\E{\widehat{\Pi}_t} - \Pi_t\right) = 0,\qquad  \lim_{\varepsilon\rightarrow 0}\lim_{t\rightarrow \infty} \left(\E{\widehat{\Lambda}_t} - \Lambda_t\right) = 0
	\end{equation}
	almost surely, where the matrices $\Pi_t\in\mathbb{R}^{n\times\nu}$ and $\Gamma_t\in\mathbb{R}^{1\times\nu}$ satisfy~\eqref{eq:IFI_regulator_equations}.
	As a consequence,
	\begin{equation}
	\lim_{\varepsilon\rightarrow 0}\lim_{t\rightarrow \infty}\left(\widetilde{\Pi}^x_t - \Pi_t\right) = \lim_{\varepsilon\rightarrow 0}\lim_{t\rightarrow \infty}\left(\widetilde{\Pi}^z_t - \Pi_t\right) = 0
	\end{equation} 
	holds almost surely as well. Then the result follows as in the proof of Theorem~\ref{theorem:AFI_solution}.
\end{proof}

\begin{remark}
If $\omega$ is not available, but can be measured through the output $y_t^a$, it is possible to incorporate an observer for $\omega$ in the regulator, alongside the observer for $x_t$. The estimate provided by this observer can then replace $\omega$ in the construction of $\Delta \widehat{W}_\varepsilon$ and in the control $u_t$.
\end{remark}

\section{Example}
\label{sec:example}
In this section we illustrate the theory by means of a numerical example. We show that it is possible to achieve approximate regulation via the hybrid scheme introduced in Sections~\ref{section:approximate_fi}, in the case of full information, and \ref{section:approximate_ef}, in the case of output feedback. In particular, we point out that, as proved in Theorems~\ref{theorem:AFI_solution} and \ref{theorem:AEF_solution}, the accuracy of the approximation increases as the exogenous input approaches zero and as the measurements are acquired with higher frequency.

Consider the electrical circuit displayed in Figure~\ref{fig:circuit}. The exogenous signal $\tilde \omega$ is a combination of two voltage harmonics and the aim is to regulate the voltage on a resistive load $R_L$ to the sinusoid with the smallest frequency. The control input is represented by a current injection and we assume that nearby electrical appliances can cause a random modification of the reactive components, \emph{i.e.} a generic reactive component $X$ is such that $X^{-1} = X_0^{-1} + X_1^{-1}\dot{\mathcal{W}}$, where $\dot{\mathcal{W}}$ is generalised white noise. See \cite{Damm2004}, \cite{Samuels1959} and \cite{Ugrinovskii1999} for more details on this way of modelling uncertain circuits.
\begin{figure}
    \centering
    \includegraphics[width=\columnwidth]{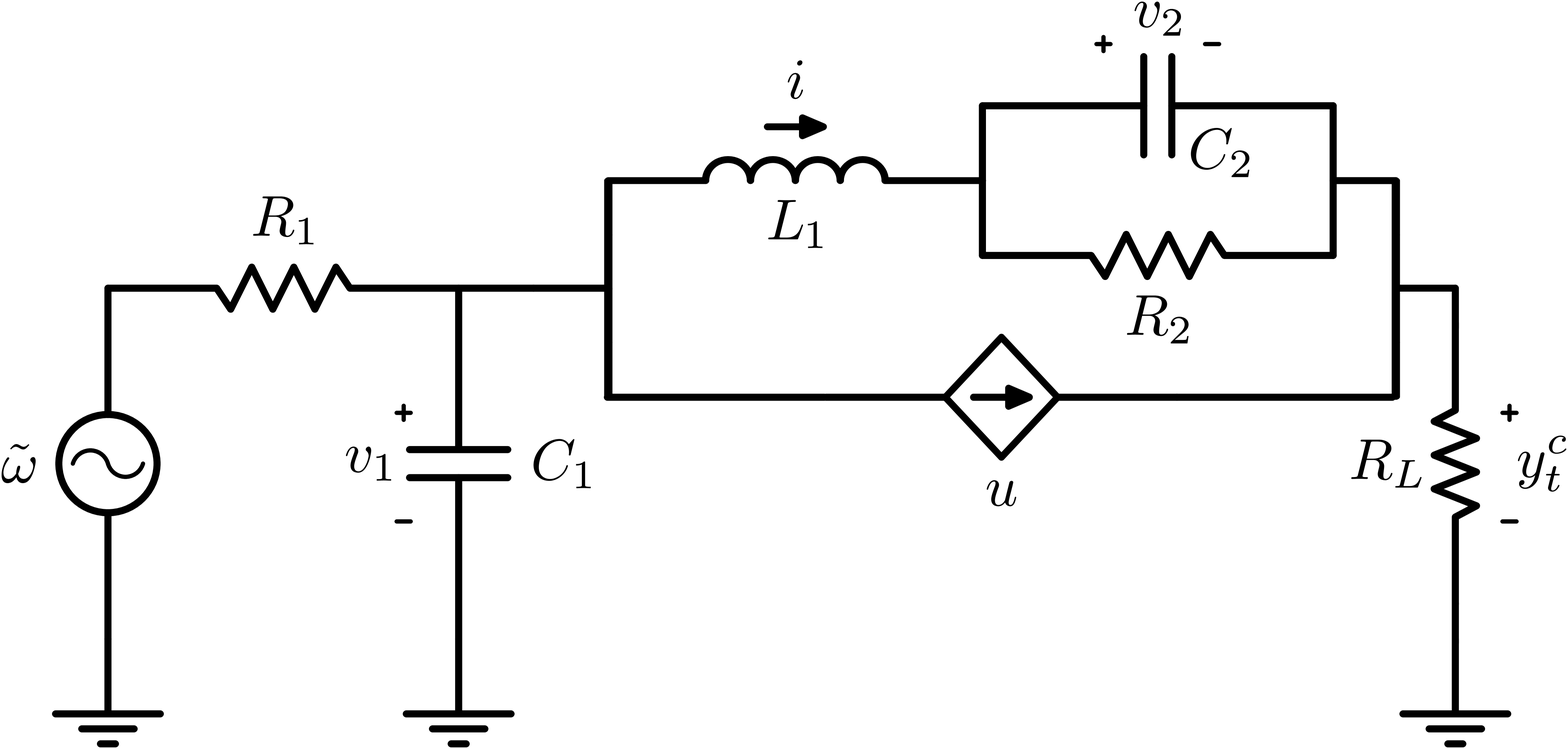}
    \caption{Uncertain electrical circuit.}
    \label{fig:circuit}
\end{figure}

The exogenous system has a matrix $S~=~2\pi\;\text{diag}\,(0,10\,S_0,50\,S_0)$, where
\begin{equation}
    S_0 = \begin{bmatrix}
    0 & 1 \\ -1 & 0\end{bmatrix}
\end{equation}
and $\tilde \omega = P_S \omega $, with $P_S = \begin{bmatrix}1 & 1 & 0 & 1 & 0 \end{bmatrix}$. We select the initial condition $\omega_0 = d\begin{bmatrix}a_0 & a_1 & 0 & a_2 & 0 \end{bmatrix}^\top$, with $a_0 = 60$, $a_1 = 5$, $a_2 = 1$ and $d\in\mathbb{R}$ a free parameter. In this way, the matrix $P_S$ selects the DC and the cosinusoidal components of the vector $\omega$, thus having $\tilde \omega(t) = d(a_0 + a_1\cos(2\pi \, 10t) + a_2 \cos(2\pi\,50t))$.  Setting the state as $x_t = \begin{bmatrix}
i, v_1, v_2
\end{bmatrix}$, the matrices of the system are
\begin{equation}
\begin{aligned}
    A = \begin{bmatrix}
    -\frac{R_L}{L_1} & \frac{1}{L_1} & -\frac{1}{L_1}\\
    -\frac{1}{C_1} & -\frac{1}{R_1C_1} & 0 \\
    \frac{1}{C_2} & 0 & -\frac{1}{R_2C_2}
    \end{bmatrix}, \quad B = \begin{bmatrix}
    -\frac{R_L}{L_1} \\ -\frac{1}{C_1} \\ 0
    \end{bmatrix},\\
    P = \begin{bmatrix}
    0 \\ \frac{1}{R_1C_1} \\ 0
    \end{bmatrix}P_S, \; C = \begin{bmatrix}
    R_L & 0 & 0
    \end{bmatrix}, \; D = R_L.
\end{aligned}
\end{equation}
The values of the parameters are selected as $R_1 = 1 \Omega$, $R_2 = 4\Omega$, $R_L = 20\Omega$, $C_1 = 10$mF, $C_2 = 20$mF, $L_1 = 200$mH. We assume that the uncertainty induced on the reactive components has a standard deviation of the $1\%$ of their nominal values, that is $F = 0.01A$, $G = 0.01B$ and $R = 0.01P$. Our aim is to replicate on the load the harmonic at $10$Hz and to cancel the harmonic at $50$Hz. Therefore, $Q = -\begin{bmatrix}1 & 1 & 0 & 0 & 0 \end{bmatrix}$. The initial condition $x_0$ has been set to zero. The matrix $K$ has been chosen as 
   $K=\begin{bmatrix}
        -0.07 & 0.04 & 0.06
    \end{bmatrix}$ 
and used both as the full-information and as the output-feedback gain. In the output-feedback case we assume that the current $i$ and the voltage $v_1$ are measured, \emph{i.e.}
    $C_a = \begin{bmatrix}
    1 & 0 & 0
    \end{bmatrix}$, 
    $C_b = \begin{bmatrix}
    0 & 1 & 0
    \end{bmatrix}$.
The gain $L(t)$ has been chosen piecewise constant as suggested in Remark~\ref{remark:choice_of_gains}. This choice of $K$ and $L$ is such that the closed-loop system, in both the full-information and output-feedback cases, are asymptotically stable and \textbf{(EC)} is satisfied.

The discrete-time numerical implementation of the hybrid controller has required an integration method involving two different sampling periods: 1) $\varepsilon$ is the sampling period at which the compensations for the diffusion term have been performed; 2) a smaller sampling period ($5\cdot 10^{-7}$)  has been used to simulate the continuous-time dynamics via a forward-Euler scheme.

First, we show that $e^{ss}_t$ decreases as $\omega_0$ approaches zero. To do so, we fix the sampling period $\varepsilon = 5\cdot 10^{-5}$ and we perform three simulations setting $d = 10$, $d = 1$ and $d = 0.1$, respectively. Figures \ref{fig:FI_w0} and \ref{fig:EF_w0} show the time history of the tracking error $e_t$ in the full-information and output-feedback cases, respectively. The insets show the detail when the initial transient response has vanished. The plots confirm that as $d$ is decreased, hence $\Vert\omega_0\Vert$ is decreased, tracking is improved. 
\begin{figure}
    \centering
    \includegraphics[width=1\columnwidth]{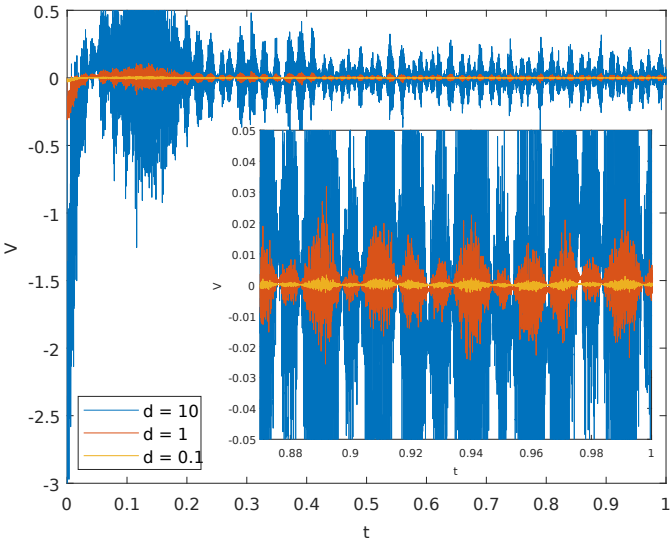}
    \caption{Time history of the tracking error $e_t$ for different values of $d$ and $\varepsilon = 5\cdot 10^{-5}$, in the case of full information. Namely: $d=10$ (blue line), $d=1$ (orange line) and $d=0.1$ (yellow line). Inset: detail at steady state.}
    \label{fig:FI_w0}
\end{figure}

\begin{figure}
    \centering
    \includegraphics[width=1\columnwidth]{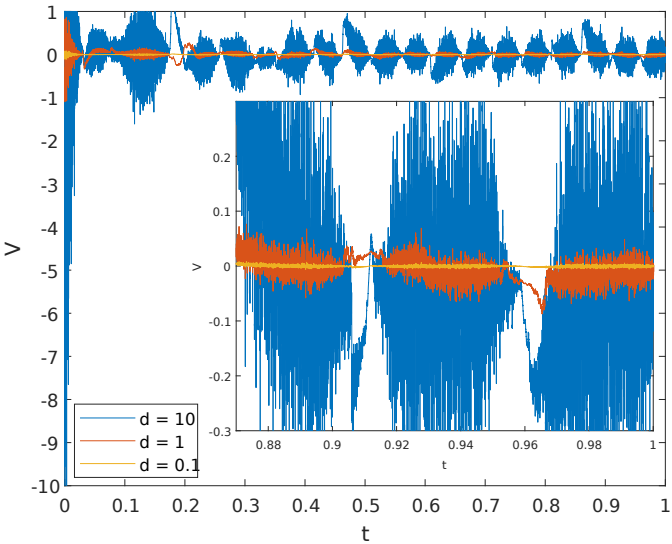}
    \caption{Time history of the tracking error $e_t$ for different values of $d$ and $\varepsilon = 5\cdot 10^{-5}$, in the case of output feedback. Namely: $d=10$ (blue line), $d=1$ (orange line) and $d=0.1$ (yellow line). Inset: detail at steady state.}
    \label{fig:EF_w0}
\end{figure}
Second, we show that $e^{ss}_t$ decreases as $\varepsilon$ approaches zero as well. We fix $d=2$ and we perform three simulations setting $\varepsilon = 5\cdot 10^{-4}$, $\varepsilon = 5\cdot 10^{-5}$ and $\varepsilon = 5\cdot 10^{-6}$, respectively. An additional simulation, where regulation in the mean sense is achieved, \emph{i.e.} jump corrections never happen (equivalently, $\varepsilon = +\infty$), has been carried out. Figures \ref{fig:FI_eps} and \ref{fig:EF_eps} show the time history of the tracking error $e_t$ in the full-information and output-feedback cases, respectively. The insets show the detail when the initial transient response has vanished. The plots confirm that a smaller sampling period improves the steady-state tracking.

\begin{figure}
    \centering
    \includegraphics[width=0.95\columnwidth]{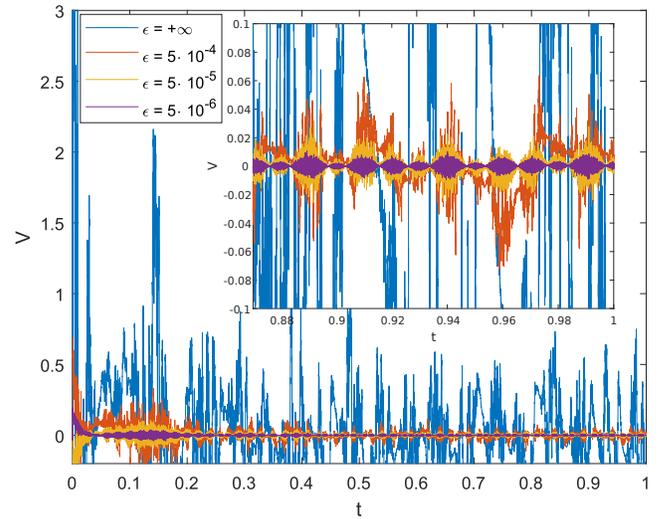}
    \caption{Time history of the tracking error $e_t$ for different values of $\varepsilon$ and $d = 2$, in the case of full information. Namely: $\varepsilon=+\infty$ (blue line), $\varepsilon = 5\cdot 10^{-4}$ (orange line), $\varepsilon = 5\cdot 10^{-5}$ (yellow line) and $\varepsilon = 5\cdot 10^{-6}$ (purple line). Inset: detail at steady state.}
    \label{fig:FI_eps}
\end{figure}

\begin{figure}
    \centering
    \includegraphics[width=0.95\columnwidth]{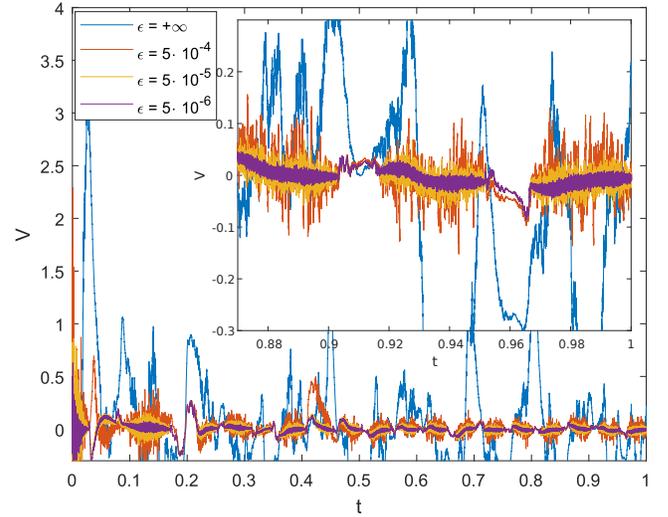}
    \caption{Time history of the tracking error $e_t$ for different values of $\varepsilon$ and $d = 2$, in the case of output feedback. Namely: $\varepsilon=+\infty$ (blue line), $\varepsilon = 5\cdot 10^{-4}$ (orange line), $\varepsilon = 5\cdot 10^{-5}$ (yellow line) and $\varepsilon = 5\cdot 10^{-6}$ (purple line). Inset: detail at steady state.}
    \label{fig:EF_eps}
\end{figure}

\begin{remark}
Assume $z_t \approx x_t$. If $v^z(k-1) \approx C_b(Fx_{t_{k-1}} + Gu_{t_{k-1}} + R\omega(t_{k-1}))$ is close to zero, the noise affecting the system gives a negligible contribution to the dynamics of the output, yet possibly affecting the dynamics of the state. This has a practical implication. In fact, if we had infinite machine precision, a very small $v^z(k-1)$ would still lead to a good \emph{a-posteriori} estimation of the Brownian motion increment in the interval $[t_{k-1}, t_k)$. However, approximation errors cause a considerable mismatch between the true and the estimated increments, thus compromising the integration of both the regulator equations and the state observer. Therefore, when implementing the output-feedback control architecture, it is beneficial not to perform any correction at time $t_k$ when $\vert v^z(k-1)\vert$ is below a predefined threshold. This is equivalent to setting $\Delta\E{\widehat{W}_\varepsilon}(k) = 0$. When this happens, the performances of the controller slightly worsen, but they improve as soon as $\vert v^z\vert$ is above the threshold. The effects of this scheme on the tracking is visible in Figures \ref{fig:EF_w0} and \ref{fig:EF_eps}, where periodical sudden variations, yet small in norms, can be observed in the time history of the tracking error.
\end{remark}

\section{Conclusions}
\label{sec:conclusions}
In this paper we have defined and solved the full-information and output-feedback output regulation problems for a general class of linear stochastic systems. In particular, we have shown that the exact integration of the regulator equations requires access to the Brownian motion. This hypothesis is obviously not practically sound. Therefore, we have formulated and solved approximate full-information and output-feedback problems via hybrid schemes. Namely, sampled measurements of the state or of the output of the system have been employed to estimate \emph{a posteriori} the Brownian motion increments between sampling times. Such estimates have been used to integrate the regulator equations and to synthesise a hybrid state observer. It has been shown that such solutions, though approximate, tend to the ideal counterparts as long as samples are acquired with increasing frequency. A numerical example has been provided to show the validity of the theory.
\bibliographystyle{IEEEtran}
\bibliography{bibliography}

\end{document}